\documentclass[10pt,journal,cspaper,compsoc]{IEEEtran}
\usepackage[usenames,dvipsnames]{color}
\usepackage{subfigure}
\usepackage{graphicx}
\usepackage{amsmath}
\usepackage{color}
\usepackage{multicol}
\usepackage{amssymb}
\usepackage{amsthm}

\usepackage{amsfonts}%
\usepackage{verbatim}%
\usepackage{enumerate}%
\usepackage{psfrag}
\usepackage{epsfig}
\usepackage{multicol}
\usepackage{setspace}
\usepackage{dsfont}
\usepackage{balance}
\usepackage{cite}
\usepackage{float}
\newtheorem{theorem}{Theorem}
\newtheorem{proposition}{Proposition}

\ifCLASSOPTIONcompsoc
\else
\fi
\ifCLASSINFOpdf
\else
\fi
\begin{document}
\title{On the Design of Relay--Assisted Primary--Secondary Networks}
\author{Ahmed~El~Shafie,~\IEEEmembership{Member,~IEEE,}
        Tamer~Khattab,~\IEEEmembership{Member,~IEEE,}
        Ahmed~Sultan,~\IEEEmembership{Member,~IEEE}% <-this % stops a space
        ~and H. Vincent Poor,~\IEEEmembership{Fellow,~IEEE}
\thanks{A. El Shafie is with Wireless Intelligent Networks Center (WINC), Nile University, Giza, Egypt. He is also with Electrical Engineering, Qatar University, Doha, Qatar (e-mail: ahmed.salahelshafie@gmail.com).}% <-this % stops a space
\thanks{T. Khattab is with Electrical Engineering, Qatar University, Doha, Qatar (email: tkhatta@ieee.org).}% <-this % stops a space
\thanks{A. Sultan is with Electrical Engineering, Alexandria University, Alexandria, Egypt (email: salatino@stanfordalumni.org).}
\thanks{H. V. Poor
is with the Department of Electrical Engineering, Princeton University, Princeton, NJ 08544 USA (email:
poor@princeton.edu).}% <-this % stops a space
\thanks{This research work is supported by Qatar National Research Fund (QNRF) under grant number NPRP 09-1168-2-455.}}

\markboth{}%
{El~Shafie \MakeLowercase{\textit{et al.}}: On the Design of Relay--Assisted}



\date{}
\maketitle

%\vspace{-1.5 cm}

\begin{abstract}
%\boldmath
The use of $N$ cognitive relays to assist primary and secondary transmissions in a time-slotted cognitive setting with one primary user (PU) and one secondary user (SU) is investigated. An overlapped spectrum sensing strategy is proposed for channel sensing, where the SU senses the channel for $\tau$ seconds from the beginning of the time slot and the cognitive relays sense the channel for $2 \tau$ seconds from the beginning of the time slot, thus providing the SU with an intrinsic priority over the relays. The relays sense the channel over the interval $[0,\tau]$ to detect primary activity and over the interval $[\tau,2\tau]$ to detect secondary activity. The relays help both the PU and SU to deliver their undelivered packets and transmit when both are idle. Two optimization-based formulations with quality of service constraints involving queueing delay are studied. Both cases of perfect and imperfect spectrum sensing are investigated. These results show the benefits of relaying and its ability to enhance both primary and secondary performance, especially in the case of no direct link between the PU and the SU transmitters and their respective receivers. Three packet decoding strategies at the relays are also investigated and their performance is compared.
\end{abstract}
\begin{IEEEkeywords}
Cognitive Radio, Queueing Delay, Relaying, Cooperative Communications, Stability Analysis.
\end{IEEEkeywords}

% \fi
\IEEEpeerreviewmaketitle

%\vspace{-0.5 cm}

\section{Introduction}
In the quest for efficient usage of radio spectrum, high reliability and high speed wireless transmission, cognitive radio and cooperative communications emerge as two of the most promising technologies. In cooperative communications \cite{laneman2004cooperative,sadek2007multinode,liu2009cooperative}, a portion of the
channel resources are assigned to one or more relays for
cooperation. These relays cooperate with a source node to help
in forwarding its data to a destination. This enhances communication reliability, reduces the required transmitted power and achieves spatial diversity.
The use of relays may, however, result in some bandwidth efficiency loss
because of the channel resources assigned to the relays to
perform their task. A cognitive relay is a viable solution to this problem as the relay utilizes the channel only when the source nodes are idle, i.e. not utilizing the spectrum.

\cite{sadek} considers a cognitive relay that aids multiple nodes in transmitting their data to a common receiver. The proposed protocol exploits
source burstiness to enable cooperation during silence periods of different nodes in a time-division multiple access (TDMA) network. In \cite{krikidis2009cognitive}, Krikidis {\it et al.} proposed to deploy a {\it dumb} relay node in cognitive radio
networks to improve
network spectrum efficiency. The relay aids both the primary and the secondary users. The proposed protocol is analyzed
and optimized for a network model consisting of a pair of primary users (PUs) and a pair of secondary users (SUs). In \cite{el2011opportunistic}, multiple relays serve multiple PUs during their silence periods. A number of secondary links coexist with the system and two secondary access scenarios are investigated. Under the first, the SUs are able to sense the activity of both the PUs and the relays, thereby remaining silent when any of them is active. In the second scenario, the relays and the SUs randomly access the channel and their transmissions may collide.

Relays with buffers are also considered in
\cite{xia2008buffering,Zlatanov,ikhlef2011buffers,krikidis2012buffer,liu2011delay,cloose}. The max-max relay selection policy is considered in \cite{ikhlef2011buffers}. Buffered relays enable the selection of the relay with the best source-relay
channel for reception and the best relay-destination channel
for transmission. The scheme relies on a two-slot protocol where the schedule for the source and relay transmission is fixed a priori. This limitation is relaxed in \cite{krikidis2012buffer} where each slot is allocated dynamically to the source or relay transmission according to the instantaneous quality of the links and the state of the buffers. In \cite{liu2011delay} and \cite{cloose}, the authors considered two-hop communication, where the SU exploits periods of silence of the PU to transmit its packets to a set of relays. Moreover, the relays can transmit even when the PU is busy because they can act together and create a beamformer to suppress or even null the interference at the primary receiver. The instantaneous channel gains are assumed to be known at the relay stations.

In this work, we consider buffered relays with cognitive capabilities. The relays serve two users with different priorities: a PU and an SU.\footnote{The proposed cognitive cooperation protocols and the theoretical development
in this paper can be generalized to cognitive radio networks with more PUs and more SUs, in which the PUs and the SUs are operating under TDMA or frequency-division multiple access (FDMA).}  The relays accept a fraction of the undelivered primary and secondary packets into their buffers and forward these packets to the primary and secondary destinations. We do not assume instantaneous channel knowledge and, hence, our protocol does not involve relay selection on the basis of instantaneous channel quality. We propose a particular overlapped spectrum sensing scheme in order to regulate the operation of the PU, the SU and the relays.

We can summarize the contributions in this paper as follows. We consider one PU and one SU in the presence of $N$ cognitive relays. The relays are used to help both the PU and the SU in communicating their data packets to their respective receivers. We propose a novel overlapped spectrum sensing technique to coordinate channel access. More specifically, the SU senses the channel for $\tau$ seconds from the beginning of the time slot to detect possible activity of the PU, while the relays sense the channel for $2 \tau$ seconds from the beginning of the time slot. Each relay senses the channel over the interval $[0,\tau]$ to detect possible activity of the PU and over the interval $[\tau,2\tau]$ to detect the activity of the SU. The SU transmits a packet from its queue if the PU is sensed to be idle. For the relays to transmit, they must sense both the PU and the SU to be inactive. We investigate three strategies for the decoding of the primary and the secondary transmissions at the relays. We propose an ordered acceptance strategy, denoted by $\mathcal{S}_{\rm OD}$, in which the relays are ordered in terms of accepting the undelivered packets of the PU and the SU into their queues. To simplify the decoding process, we propose random assignment decoding, denoted by $\mathcal{S}_{\rm RD}$, and round robin decoding, denoted by $\mathcal{S}_{\rm RR}$, in which each relay is assigned to the decoding role for a fraction of the time slots. We study the optimal secondary average service rate given certain average arrival rates to the primary and the secondary queues. Also, we investigate the minimum number of relays needed to achieve a specific level of quality of service (QoS) for the users. We study the case of sensing errors at the relays' spectrum sensors. In contrast with many works involving automatic repeat request (ARQ) feedback, we take into account the cost of the feedback duration, which is a throughput loss as the time allowed for actual data transmission is reduced. Finally, in Appendix A, we provide a proof of the advantage of $\mathcal{S}_{\rm OD}$ over $\mathcal{S}_{\rm RD}$ and $\mathcal{S}_{\rm RR}$, in terms of the service rates, for the case of a negligible feedback duration per relay.

The rest of the paper is organized as follows. In Section \ref{sec2}, we describe the system model adopted in this paper. The problem formulations are presented in Section \ref{problemformulation}. The system with sensing errors is investigated in Section \ref{sensingerror}. We provide some numerical results in Section \ref{num} and conclude the paper in Section \ref{conc}.

%%%%%%%%%%%%%%%%%%%%%%%%%%%%%%%%%%%%%
%%%%%%%%%%%%%%%%%System Model
%%%%%%%%%%%%%%%%%%%%%%
%%%%%%%%%%%%%%%%%%%%%%%%

\section{System Model}\label{sec2}

The network consists of one primary transmitter `${\rm p}$', one secondary transmitter `${\rm s}$', one primary destination `${\rm pd}$, one secondary destination `${\rm sd}$', and a set of $N$ relays labeled as $1,2,3,\dots,N$ as shown in Fig. \ref{fig1}. The relays are half-duplex, which means that they either transmit or receive but cannot do both at the same time. We consider a wireless collision channel model where concurrent transmissions by two or more nodes are assumed to be lost. Each of the PU and the SU has an infinite buffer for storing fixed-length packets. Each terminal operates as a discrete-time Geo/Geo/1 queue \cite{alfa2010queueing}.\footnote{The notion of discrete-time Geo/Geo/1 queue is used to describe a queueing
system with a Bernoulli arrival process and geometrically distributed service
times.} The
arrivals at the primary and secondary queues are independent and identically
distributed (i.i.d.) Bernoulli random variables from slot to slot with means $\lambda_{\rm p}\in[0,1]$ and $\lambda_{\rm s}\in[0,1]$ packets per time slot, respectively. Arrival processes at the primary and secondary buffers are statistically independent of one another. Each relay has two queues: a queue for relaying the primary packets denoted by $Q_{{\rm p},k}$, and a queue for relaying the secondary packets denoted by $Q_{{\rm s},k}$, where $k\in\{1,2,\dots,N\}$. The relays help both the PU and the SU to deliver their packets in the periods of silence of both of them. If a terminal transmits during a time slot, it sends exactly one packet to its respective receiver.

We propose an overlapped spectrum sensing scheme as depicted in Fig. \ref{time_slot}. The SU senses the channel from the beginning of the time slot up to $\tau$ seconds relative to the beginning of the time slot, while all the relays sense the channel over the interval $[0,\tau]$ to detect primary activity and over the interval $[\tau,2\tau]$ to detect secondary activity. If the channel is sensed to be free over both intervals, then all the relays remain idle during the rest of the time slot except the relay that is scheduled for transmission provided that its queues are nonempty. If either the PU or the SU is sensed to be active, then the relays may switch to the receiving mode depending on the decoding strategy as explained later in Subsection \ref{MAC}. There is a feedback phase at the end of the time slot to indicate the status of packet delivery.
\begin{figure}
  \centering
  \includegraphics[width=1\columnwidth]{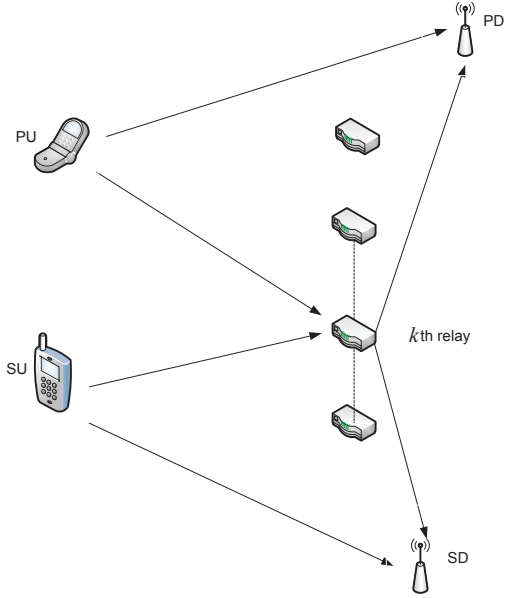}\\
  \caption{System model: $N$ relays assist primary and secondary transmissions from the PU and the SU to the primary destination (PD) and secondary destination (SD), respectively.}\label{fig1}
\end{figure}

\subsection{Medium access control (MAC) Layer}\label{MAC}

The PU transmits the packet at the head of its queue starting at the beginning of the time slot. If the PU is sensed to be idle by the SU, the SU transmits the packet at the head of its queue after $\tau$ seconds. A relay with nonempty queues transmits during a time slot after $2\tau$ seconds if it is scheduled to transmit and it senses the PU and the SU to be idle. The probability that relay $k$ is scheduled to transmit during a time slot is $\omega_k$. This means that over a large number of time slots relay $k$ is assigned to transmit during a fraction $\omega_k$ of the total time slots. It is clear that $\sum_{k=1}^N \omega_k = 1$. We define a vector $\boldsymbol{\omega}=[\omega_1,\omega_2,\dots,\omega_N]$ to indicate the fraction of time slots allocated to each relay for transmission. If relay $k$ is scheduled for transmission, which occurs with probability $\omega_k$, it chooses a packet from $Q_{{\rm p},k}$ with probability $\alpha_{k}$ and from $Q_{{\rm s},k}$ with probability $1-\alpha_{k}$. We define the vector $\boldsymbol{\alpha}=[\alpha_1,\alpha_2,\dots,\alpha_N]$.

If a relay receives during a time slot, it distinguishes between the primary and the secondary transmissions through an identifier contained in each transmitted packet.\footnote{The reason the origin of a packet is identified by a certain embedded identifier is that spectrum sensing may be erroneous. If sensing were perfect, the relay could identify the origin of transmission depending on whether it has proceeded at the beginning of the time slot or after $\tau$ seconds. In this paper, although we start with the perfect sensing case, we later address the issue of spectrum sensing errors.} If a relay correctly receives a packet, it decides to accept it with a certain probability. The acceptance probability vector of the undelivered primary packets is $\boldsymbol{f}_{\rm p}$, where the element $f_{{\rm p},k}$ is the probability that relay $k$ admits a correctly received primary packet to $Q_{{\rm p},k}$. Similarly, the vector $\boldsymbol{f}_{\rm s}$ has $N$ elements with $f_{{\rm s},k}$ being the probability of admitting a correctly received secondary packet to $Q_{{\rm s},k}$.

\begin{figure}
\centering
  \includegraphics[width=1\columnwidth]{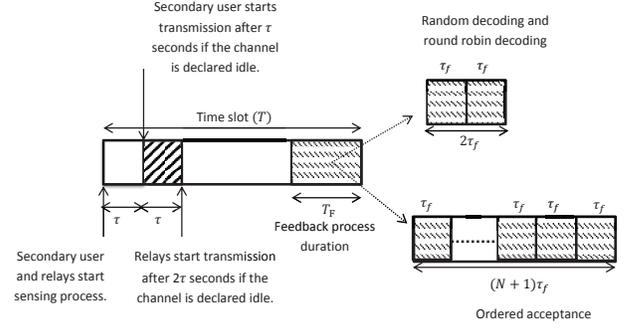}\\
  \caption{Time slot structure. The explanation for the feedback process and its duration is provided at the end of Section \ref{sec2}.
 }\label{time_slot}
\end{figure}
\subsubsection{Ordered Acceptance Strategy}
Under the ordered acceptance strategy, $\mathcal{S}_{\rm OD}$, if a relay senses either the PU or SU to be busy, it operates in the receiving mode till the transmission time within the time slot is over. If the primary destination (PD) or secondary destination (SD) acknowledges the correct reception of the transmitted packet by sending an acknowledgment (ACK) message, the relays discard what they have received from the PU or the SU. If the PD or SD declares its failure to decode the received packet correctly by generating a negative acknowledgment (NACK) message, the relays attempt to decode the received packet and determine its origin. If the received packet is correctly decoded and, hence, its origin is identified by the first-ranked relay, it decides whether to accept the packet. If the packet is admitted, an ACK is transmitted by the relay to inform the PU or SU to drop the packet from its queue and to notify the other relays that the packet has already been accepted. If the first-ranked relay receives the packet in error or does not accept it, it remains silent and the second-ranked relay makes the acceptance decision in case it has received the packet correctly. Generally, a relay, depending on its decoding rank, decides whether to accept a correctly decoded packet provided that all the preceding relays do not admit the packet. We assume perfect decoding of the feedback messages at all nodes. This assumption is reasonable when strong channel codes with low modulation indices are employed for the feedback channel \cite{sadek}.

The relays' acceptance order is the $N$-tuple $\boldsymbol{m}_n = (m_1,m_2,\dots, m_i, \dots, m_N)$, where $m_i\in\{1,2, \dots, N\}$ and $m_i \neq m_k , \forall i,k : i \neq k$. The $N$-tuple $(m_1,m_2,\dots, m_N)$ means that relay $1$ is assigned the $m_1$th acceptance rank, relay $2$ is assigned the $m_2$th rank and so on. It is evident that $\boldsymbol{m}_n$ is a permutation, $\Pi_n$, over the set $\{1, 2, ...,N\}$ and there are $N!$ such permutations,  where $N!$ indicates the factorial of $N$. We define the probability $\rho_n^{\left(\rm p\right)}$ as the probability of the $j$th permutation, $\Pi_n$, resulting in the acceptance order $(m_1,m_2,\dots,m_N)$ if the received packet comes from the PU. This probability denotes the fraction of time slots with this ranking order. Probability $\rho_n^{\left(\rm s\right)}$ is defined in a similar fashion if the packet is transmitted by the SU. Vectors $\boldsymbol{\rho}^{\left(\rm p\right)}$ and $\boldsymbol{\rho}^{\left(\rm s\right)}$ have the aforementioned probabilities as elements and both have $N!$ elements.

The medium access control (MAC) operation can be summarized as follows:
\begin{itemize}
\item At the beginning of a time slot, the PU transmits the packet at the head of its queue to the primary receiver. Due to the broadcast nature of the wireless channel, the SU and the relays can listen to the transmitted primary packet.
\item The SU senses the channel over the first $\tau$ seconds of the time slot. If the SU detects the channel to be free from primary activity, it transmits from its queue if it is nonempty. The relays can overhear the secondary transmission.
\item If the PU is active and the transmitted packet is received correctly by the primary receiver, an ACK message is fed back from the receiver. The packet is then dropped from the primary queue. The relays also discard what they have received.
\item If the primary packet is not received correctly, a NACK message is fed back from the primary receiver. The relays then attempt to decode the received packet and determine its origin. Based on their primary packet acceptance ranking, the first-ranked relay decides whether or not to accept the primary packet if it is decoded correctly. If the packet is accepted, an ACK message is transmitted, thereby inducing the primary transmitter to drop the packet. If the first-ranked cognitive relay fails to decode the primary packet or does not accept it, the second-ranked relay tries to do so. This relay issues an ACK signal if it decodes the packet successfully and decides to accept it. This operation continues in ranking order till a relay decodes and accepts the packet. If no relay accepts the packet, it is kept in the PU's queue for retransmission.
\item In the case of secondary transmission, the relays perform the same operation as described for primary transmission. The ranking of relays to accept the secondary packets differs from the ranking of accepting the primary packets.
\item If both the PU and the SU are found to be idle, the relays start transmitting the packets at the heads of their queues. The fraction of time slots in which relay $k$ is scheduled to transmit is $\omega_k$.
\end{itemize}

\subsubsection{Random Assignment Decoding and Random Round Robin Strategies}
The difference between random assignment decoding, $\mathcal{S}_{\rm RD}$, and $\mathcal{S}_{\rm OD}$ is that in $\mathcal{S}_{\rm RD}$ only one relay is scheduled to decode, and possibly accept, the undelivered primary or secondary packet at any slot. The probability that relay $k$ is assigned the decoding role in a time slot is denoted as $\beta_k$. We define the vector $\boldsymbol{\beta}=\bigg[\beta_1,\beta_2,\dots,\beta_N\bigg]$ with the constraint $\sum_{k=1}^N \beta_k = 1$. The vectors $\boldsymbol{\alpha}$, $\boldsymbol{\omega}$, $\boldsymbol{f}_{\rm p}$ and $\boldsymbol{f}_{\rm s}$ are similar to those in $\mathcal{S}_{\rm OD}$. The operation of the relays can be summarized as follows:
\begin{itemize}
\item At the beginning of each time slot, the index, $k$, of the randomly selected relay is generated according to the discrete distribution $\boldsymbol{\beta}$.
\item If the primary packet is not received correctly, a NACK message is fed back from the primary receiver. The relay that is assigned for packet decoding tries to decode the undelivered primary packet. If the packet is decoded correctly, it is accepted with a certain probability and an ACK message is transmitted, thereby inducing the PU to drop the packet. If the cognitive relay assigned for decoding fails to decode the primary packet or does not accept it, the packet is kept in the primary queue for retransmission.
    \item If the PU is sensed by the SU to be idle, the SU, if its queue is not empty, starts transmission and the relays repeat the same operation as described earlier for the PU's relaying scenario. Recall that the origin of the received packets at the assigned relay can be known from the packet's identifier. Based on the identifier, assigned relay $k$ accepts the packet into $Q_{{\rm p},k}$ with probability $f_{{\rm p},k}$, or into $Q_{{\rm s},k}$ with probability $f_{{\rm s},k}$.
\end{itemize}

Random round robin decoding, $\mathcal{S}_{\rm RR}$, is a simplification of $\mathcal{S}_{\rm RD}$ in which the decoding assignment probability is equal for all relays. That is, $ \beta_k=1/N$, where $k\in\{1,2,\dots,N\}$.
%%%%%%%%%%%%%%%%%
%%%%%%%%%%%%%%%%%%%
%%%%%%%%%%%%%%%%%%%%

%%%%%%%%%%%%%%%%%%%%%%%%%%%%%
%%%%%%%%%%%%%%%%%%%%%%%%%%%%%%
%%%%%%%%%%%%%%%%%%%%%%%%%%%%%%%
%% Physical Layer
%%%%%%%%%%%%%%%%%%%%%%%%%
%%%%%%%%%%%%%%%%%%%%%%%%%%
%%%%%%%%%%%%%%%%%%%%%%%

\subsection{Physical (PHY) Layer}
The channel outage event for the relays and the SU can be calculated as follows. The transmitters adjust their transmission rates depending on when they start transmission during the time slot.
Assuming that the number of bits in a packet is $b$ and the time slot duration is $T$, the transmission rate is
\begin{equation}
r_i=\frac{b}{T\left(1-\frac{T_{\rm SF}}{T}\right)}.
\label{r_i}
\end{equation}
\noindent with $T_{\rm SF}=i \tau+T_{\rm F} < T$, where $T_{\rm F}$ is the time needed to execute the feedback process. The parameter $i=0$ if the transmitter is the PU as transmission proceeds at the very beginning of the time slot, $i=1$ for the SU as its transmission is preceded by a spectrum sensing period of $\tau$ seconds, and $i=2$ for all relays because their transmission is preceded by a spectrum sensing period of $2 \tau$ seconds relative to the beginning of the time slot. Outage of a link occurs when the transmission rate exceeds the channel capacity. Hence, the outage probability of the link between node $j$ and node $k$ is given by \cite{sadek}
\begin{equation}
P_{j,k}={\rm Pr}\biggr \{r_i > W \log_{2}\left(1+\gamma_{j,k} h_{j,k}\right)\biggr\}
\end{equation}
\noindent where $W$ is the bandwidth of the channel, $\gamma_{j,k}$ is the received signal to noise ratio (SNR) when the channel gain is equal to unity, and $h_{j,k}$ is the channel power gain, which is exponentially distributed in the case of Rayleigh fading. The channel gain, $h_{j,k}$, is assumed to be independent from slot to slot and link to link. The outage probability can be written as
\begin{equation}
P_{j,k}={\rm Pr}\Big\{h_{j,k}<\frac{2^{\frac{r_i}{W}}-1}{\gamma_{j,k}}\Big\}
\label{choutage}
\end{equation}
\noindent Assuming that the mean value of $h_{j,k}$ is $\sigma_{j,k}$,
\begin{equation}
P_{j,k}=1-\exp\bigg(-\frac{2^{\frac{r_i}{W}}-1}{\gamma_{j,k}\sigma_{j,k}}\bigg).
\label{outprob}
\end{equation}
\noindent Let $\overline{P}_{j,k}=1-P_{j,k}$ be the probability of correct reception. It is therefore given by
\begin{equation}\label{correctreception}
\overline{P}_{j,k}=\exp\bigg(-\frac{2^{\frac{b}{TW\left(1-\frac{T_{\rm SF}}{T}\right)}}-1}{\gamma_{j,k}\sigma_{j,k}}\bigg).
\end{equation}

Note that the duration of the feedback process, $T_{\rm F}$, varies according to the strategy in the decoding that the relays adopt. In the case of $\mathcal{S}_{\rm OD}$, the relays are ordered in terms of sending the ACK messages if they accept a correctly received packet. If each relay needs $\tau_{\rm f}$ seconds, then the overall feedback duration is $T_{\rm F}=\left(N+1\right)\tau_{\rm f}$ given that the PD also needs $\tau_{\rm f}$ seconds to acknowledge the reception of a data packet. On the other hand, $\mathcal{S}_{\rm RD}$ and $\mathcal{S}_{\rm RR}$ need only $T_{\rm F}=2\tau_{\rm f}$ for the feedback process to be executed. The increase in the feedback duration is interpreted as an increase in the outage probability of the channels. This fact can be seen easily from (\ref{correctreception}). It should be mentioned that the decoding process, taking into account the feedback duration, will cause a reduction in the allowable data transmission time of both the primary and secondary transmissions, i.e., the total transmission time will be reduced to $T-T_{\rm SF}$.

%%%%%%%%%%%%%%%%%%%%%%%%%%%%
%%%%%%%%%%%%%%%%%%%%%%%%%%%%%
%%%%%%%%%%%%%%%%%%%%%%%%%%%%
% PROBLEM FORMULATION
%%%%%%%%%%%%%%%%%%%%%%%%%%%%%
%%%%%%%%%%%%%%%%%%%%%%%%%%%%%
%%%%%%%%%%%%%%%%%%%%%%%%%%%%%
%%%%%%%%%%%%%%%%%%%%%%%%%%%%%
\section{Problem Formulations}\label{problemformulation}

We start the performance analysis of the different protocols with the case of perfect spectrum sensing; then we consider spectrum sensing errors in the next section.

\subsection{Average Arrival and Service Rates}
\subsubsection{Ordered Acceptance}
Consider the primary queue first. A packet can be served in either one of the following events: the channel between PU and PD is not in outage; or the primary channel is in outage and one of the relays decodes correctly and accepts the packet. Note that for relay $k$ to get the primary packet, all the relays having a higher priority in accepting the packet should either fail to receive the packet correctly due to channel outage, or reject the packet. Hereinafter, we adopt the notation $\overline{x}=1-x$. The average service rate of the primary queue is given by
\begin{equation}
\small \begin{split}
  \mu_{\rm p}&=\overline{P}_{\rm p,pd}+P_{\rm p,pd}\sum_{k=1}^{N} \biggr[\overline{P}_{{\rm p},k} f_{{\rm p},k}  \\& \,\ \!\times\! \sum_{\Pi_n}\rho_n^{\left(\rm p\right)} \!\prod_{\substack{{v=1} \\{v \ne k} \\ {m_v<m_k}}}^{N}\overline{\overline{P}_{{\rm p},v} f_{{\rm p},v}}\biggr].
  \end{split} \normalsize
  \label{mupordered}
\end{equation}

The secondary queue can be analyzed in a similar fashion. The probability of the primary queue being empty is \cite{Amr,sadek,rao1988stability}
\begin{equation}
\pi_{{\rm p},\circ}=1-\frac{\lambda_{\rm p}}{\mu_{\rm p}}.
\end{equation}
\noindent When the primary queue is empty, a packet in the secondary queue can be served in either one of the following events: the channel between the SU and SD is not in outage; or the channel between the SU and the SD is in outage and one of the relays decodes and decides to accept the packet. Therefore, the average service rate of the secondary queue can be written as
\begin{equation}
\small \begin{split}
    \mu_{\rm s}&= \pi_{{\rm p},\circ} \biggr[ \overline{P}_{\rm s,sd}+P_{\rm s,sd} \sum_{k=1}^{N}  \bigg(  \overline{P}_{{\rm s},k} f_{{\rm s},k} \\& \,\ \!\times\! \sum_{\Pi_n}\rho_n^{\left(\rm s\right)} \prod_{\substack{{v=1} \\{v \ne k} \\ {m_v<m_k}}}^{N}\overline{\overline{P}_{{\rm s},v} f_{{\rm s},v}}\!\bigg)\!\biggr].
    \end{split} \normalsize
    \label{musordered}
\end{equation}
\noindent The probability that the secondary queue is empty is given by
\begin{equation}
\pi_{{\rm s},\circ}=1-\frac{\lambda_{\rm s}}{\mu_{\rm s}}.
\end{equation}

Let $\lambda_{{\rm p},k}$ and $\lambda_{{\rm s},k}$ be the arrival rates at the queues $Q_{{\rm p},k}$ and $Q_{{\rm s},k}$ of relay $k$, respectively. Note that for an arrival event to occur at $Q_{{\rm p},k}$, the primary queue should be nonempty. For an arrival event to happen at $Q_{{\rm s},k}$, the primary queue must be empty to preclude primary transmission and the secondary queue should be nonempty. The expressions for the arrival rates follow directly from (\ref{mupordered}) and (\ref{musordered}) and are given by
\begin{equation}
\small \begin{split}
    \lambda_{{\rm p},k}&\!=\!\overline{\pi}_{{\rm p},\circ} P_{\rm p,pd}  \biggr[\!  \overline{P}_{{\rm p},k} f_{{\rm p},k}   \sum_{\Pi_n}\rho_n^{\left(\rm p\right)}\!\prod_{\substack{{v=1} \\{v \ne k} \\ {m_v<m_k}}}^{N}\!\overline{\overline{P}_{{\rm p},v} f_{{\rm p},v}}\!\biggr].
    \end{split} \normalsize
\end{equation}
and
\begin{equation}
\small \begin{split}
\lambda_{{\rm s},k}&\!=\!  \overline{\pi}_{{\rm s},\circ} \pi_{{\rm p},\circ} P_{\rm s,sd}  \biggr[  \overline{P}_{{\rm s},k} f_{{\rm s},k}  \sum_{\Pi_n}\rho_n^{\left(\rm s\right)}\!\prod_{\substack{{v=1} \\{v \ne k} \\ {m_v<m_k}}}^{N}\!\overline{\overline{P}_{{\rm s},v} f_{{\rm s},v}}\!\biggr].
\end{split} \normalsize
\end{equation}

For a relay to transmit, both the primary and secondary queues should be empty. Relay $k$ transmits from $Q_{{\rm p},k}$ with probability $\alpha_k$ and from $Q_{{\rm s},k}$ with probability $1-\alpha_k$.  The average service rates, $\mu_{{\rm p},k}$ and $\mu_{{\rm s},k}$, of $Q_{{\rm p},k}$ and $Q_{{\rm s},k}$ at relay $k$, respectively, are given by
\begin{equation}
\small \begin{split}
    \mu_{{\rm p},k}&\!=\!  \omega_k \pi_{{\rm p},\circ}\pi_{{\rm s},\circ} \alpha_{k} \overline{P}_{k,{\rm pd}}, \\ \mu_{{\rm s},k}&\!=\!  \omega_k \pi_{{\rm p},\circ} \pi_{{\rm s},\circ} \big(\! 1\!-\! \alpha_{k}\!\big)\overline{P}_{k,{\rm sd}}
    \end{split} \normalsize
\end{equation}

We can upper bound the mean service rate of the primary queue as follows. The maximum service rate occurs when all relays decide to accept the primary packet each time slot, i.e., $f_{{\rm p},k}=1\ \forall k$, regardless of the decoding order distribution.\footnote{This is because, in any arbitrary slot, each relay, whatever its decoding rank, will attempt to decode the primary packet and admit it, if the lower ranked relays fail in decoding it due to channels outage.} In this case, the mean service rate of the primary node becomes the probability that one of the receiving nodes' channels is not in outage. Therefore, the maximum mean service rate of the primary queue under strategy $\mathcal{S}_{\rm OD}$ is
\begin{equation}
\small \begin{split}
  \mu_{\rm p}^{\rm \max}&=1-P_{\rm p,pd} \prod_{v=1}^{N}P_{{\rm p},v}
  \label{pmax1}
  \end{split} \normalsize
\end{equation}
where $1-P_{\rm p,pd} \prod_{v=1}^{N}P_{{\rm p},v}$ is the probability that either the PD or one of the relays decodes the primary packet correctly. Similarly, the maximum mean secondary service rate under strategy $\mathcal{S}_{\rm OD}$ is
\begin{equation}
\small \begin{split}
    \mu_{\rm s}^{\rm \max}&= \biggr[ 1-P_{\rm s,sd} \prod_{v=1} ^{N}P_{{\rm s},v}\bigg]\tilde \pi_{{\rm p},\circ}\\&\!=\!\biggr[\!1\!-\!P_{\rm s,sd} \prod_{v=1} ^{N}P_{{\rm s},v}\!\biggr] \Biggr[1\!-\!\frac{\lambda_{\rm p}}{ 1\!-\!P_{\rm p,pd} \prod_{v=1}^{N}P_{{\rm p},v}}\!\Biggr]
    \label{smax11}
    \end{split} \normalsize
\end{equation}
where $\tilde \pi_{{\rm p},\circ}=1-\frac{\lambda_{\rm p}}{ 1-P_{\rm p,pd} \prod_{v=1}^{N}P_{{\rm p},v}}$ is the probability of the primary queue being empty when $\mu_{\rm p}=\mu_{\rm p}^{\rm \max}$ which upper bounds $\pi_{{\rm p},\circ}$.

%\bigskip

%%%%%%%%%%%%%%
%%%%%%%%%%%%%%%%%%%%%%%
\subsubsection{Random Assignment Decoding}
In $\mathcal{S}_{\rm RD}$, the $k$th relay is scheduled to decode the transmitted packet with probability $ \beta_k$. Hence, the average service rates of the PU and the SU are given by
\begin{equation}
\small \begin{split}
  \mu_{\rm p}&=\overline{P}_{\rm p,pd}+P_{\rm p,pd}\sum_{k=1}^{N}   \overline{P}_{{\rm p},k} f_{{\rm p},k} \beta_k, \    \\ \mu_{\rm s}&= \pi_{{\rm p},\circ} \bigg(\overline{P}_{\rm s,sd}+P_{\rm s,sd} \sum_{k=1}^{N}  \overline{P}_{{\rm s},k} f_{{\rm s},k}\beta_k\bigg).
  \end{split} \normalsize
\end{equation}
The average arrival rates to the relaying queues are given by
\begin{equation}
\small \begin{split}
    \lambda_{{\rm p},k}&=  P_{\rm p,pd}  \  \overline{P}_{{\rm p},k} f_{{\rm p},k}  \beta_k \overline{\pi}_{{\rm p},\circ},\\ \ \lambda_{{\rm s},k}&=  P_{\rm s,sd}    \overline{P}_{{\rm s},k} f_{{\rm s},k}  \beta_k \overline{\pi}_{{\rm s},\circ} \pi_{{\rm p},\circ}.
    \end{split} \normalsize
\end{equation}
\noindent The average service rates of the relaying queues are the same as in the ordered acceptance case. The mean service rate of the primary queue is upper bounded as follows. The mean service rates of the primary queue under strategy $\mathcal{S}_{\rm RD}$ can be upper bounded as follows.
\begin{equation}
\small \begin{split}
    \mu_{\rm p}&=\overline{P}_{\rm p,pd}+P_{\rm p,pd}\sum_{k=1}^{N}   \overline{P}_{{\rm p},k} f_{{\rm p},k} \beta_k \\& \le^{\mathcal{E}} \overline{P}_{\rm p,pd}+P_{\rm p,pd}\sum_{k=1}^{N}   \overline{P}_{{\rm p},k} \beta_k,
  \end{split} \normalsize
\end{equation}
where the inequality $\mathcal{E}$ holds with equality when $f_{{\rm p},k}=1$ for all relays. Since $\overline{P}_{{\rm p},k}$ belongs to the convex set $[0,1]$ and $\beta_k\in[0,1]$ is a convex set with $\sum_{k=1}^{N}\beta_k=1$, then $\sum_{k=1}^{N} \overline{P}_{{\rm p},k} \beta_k$ is a convex hull with maximum value  located at the edges, i.e., at $\beta_k\in\{0,1\}$. Accordingly  $\sum_{k=1}^{N}   \overline{P}_{{\rm p},k} \beta_k\le \max \biggr\{\overline{P}_{{\rm p},1},\overline{P}_{{\rm p},2},\dots,\overline{P}_{{\rm p},N}\biggr\}$, and
\begin{equation}
\small \begin{split}
    \mu_{\rm p}&\le \overline{P}_{\rm p,pd}+P_{\rm p,pd}\sum_{k=1}^{N}   \overline{P}_{{\rm p},k} \beta_k \\& \le \overline{P}_{\rm p,pd}+P_{\rm p,pd}\max \biggr\{\overline{P}_{{\rm p},1},\overline{P}_{{\rm p},2},\dots,\overline{P}_{{\rm p},N}\biggr\}\\&= 1-P_{\rm p,pd}  \min \biggr\{P_{{\rm p},1},P_{{\rm p},2},\dots,P_{{\rm p},N}\biggr\}\!=\! \mu_{\rm p}^{\rm \max}
  \end{split} \normalsize
\end{equation}
%where $\min\{u_1,u_2\}=u_2$ if $u_2<u_1$ and $\max\{u_1,u_2\}=u_1$ if $u_1>u_2$.
where $\max\{\cdot\}$ and $\min\{\cdot\}$ return the maximum and the minimum of all the values present in their arguments, respectively. The maximum mean primary and secondary service rates are
\begin{equation}
\small \begin{split}
 \mu_{\rm p}^{\rm \max}\!=\!1-P_{\rm p,pd}  \min \biggr\{P_{{\rm p},1},P_{{\rm p},2},\dots,P_{{\rm p},N}\biggr\}
   \label{pmax2}
   \end{split} \normalsize
\end{equation}
and
\begin{equation}
\small \begin{split}
    \mu_{\rm s}^{\rm \max}&\!=\!\biggr[ 1\!-\!P_{\rm s,sd}  \min \biggr\{P_{{\rm s},1},P_{{\rm s},2},\dots,P_{{\rm s},N}\biggr\}\biggr] \\& \times \biggr[\!1\!-\!\frac{\lambda_{\rm p}}{ 1\!-\!P_{\rm p,pd} \min \biggr\{P_{{\rm p},1},P_{{\rm p},2},\dots,P_{{\rm p},N}\biggr\}}\!\biggr]
    \end{split} \normalsize
\end{equation}

%\bigskip
%%%%%%%%%%%%%%
%%%%%%%%%%%%%%%%%%%%%%%
\subsubsection{Round Robin Decoding}
In $\mathcal{S}_{\rm RR}$, each relay is assigned the decoding role with equal probability, i.e., $1/N$, in a cyclic manner. The expressions are thus similar to $\mathcal{S}_{\rm RD}$ with the substitution $\beta_k=1/N$. As in the previous subsection and with setting $\beta_k=1/N$, we can obtain the maximum mean service rates of the primary and secondary queues under strategy $\mathcal{S}_{\rm RR}$. The maximum mean primary and secondary service rates are
\begin{equation}
\small \begin{split}
  \mu_{\rm p}^{\rm \max}&=1-P_{\rm p,pd}\biggr(1-\frac{\sum_{v=1}^{N}\overline{P}_{{\rm p},v}}{N}\biggr)
  \label{pmax3}
  \end{split} \normalsize
\end{equation}
and
\begin{equation}
\small \begin{split}
    \mu_{\rm s}^{\rm \max}&=\biggr[ 1-P_{\rm s,sd}\biggr(1-\frac{\sum_{v=1}^{N}\overline{P}_{{\rm s},v}}{N}\biggr)\biggr]\\& \times \Biggr[1-\frac{\lambda_{\rm p}}{1-P_{\rm p,pd}\biggr(1-\frac{\sum_{v=1}^{N}\overline{P}_{{\rm p},v}}{N}\biggr)}\Biggr].
    \end{split} \normalsize
\end{equation}

\begin{theorem}\label{pro1}
The queue service rates of $\mathcal{S}_{\rm OD}$ always outperform the queue service rates of $\mathcal{S}_{\rm RD}$ and $\mathcal{S}_{\rm RR}$ for a network with $N$ relays if the feedback duration per relay is negligible.
%In the Appendix, we provide a proof for the advantage of $\mathcal{S}_{\rm OD}$ over $\mathcal{S}_{\rm RD}$ and $\mathcal{S}_{\rm RR}$ in case of perfect sensing and a negligible feedback duration per relay.
\end{theorem}
\begin{proof}
The proof for this theorem under perfect sensing and imperfect sensing (sensing errors) is presented in Appendix~A.
\end{proof}

\begin{proposition}\label{pro2}
The SU's maximum mean service rate, $\mu^{*}_{\rm s}$, for an arbitrary decoding strategy is given by
\begin{equation}
\small \begin{split}
\mu^{*}_{\rm s} = 1\!-\!\lambda_{\rm p}.
 \end{split} \normalsize
\end{equation}
\end{proposition}
\begin{proof}
Regardless of decoding strategy, the secondary average service rate can be always upper bounded by the probability of the PU's queue being empty assuming that when the PU is idle due to its empty queue, the SU can successfully transmit its packet with probability one. This can be expressed as $\mu_{\rm s} \le \pi_{{\rm p},\circ}$. Since the probability of the PU being empty is $\pi_{{\rm p},\circ}=1-\frac{\lambda_{\rm p}}{\mu_{\rm p}} \le 1-\lambda_{\rm p}$, then $\mu_{\rm s}\le 1-\lambda_{\rm p}$.
\end{proof}

\subsection{Average Queueing Delay Analysis}

Since all network queues are decoupled, the $j$th queue $Q_j$ queueing delay when it is stable, is given by \cite{gambini2007stability,sadek}
 \begin{equation}
 D_j=\frac{1-\lambda_j}{\mu_j-\lambda_j}
 \end{equation}
 where $j\in \bigg\{{\rm p},{\rm s},({\rm p},k),({\rm s},k)\bigg\}$ and $\mu_j>\lambda_j$.
 The end-to-end mean queueing delay is the average delay that any packet experiences from its arrival at the
source queue till it arrives at the destination. In our system, each packet arriving at $Q_\ell$ experiences on the average delay of $D_\ell$ time slots, where $\ell\in\{{\rm p},s\}$. Further, a packet has an additional delay $D_{\ell,k}$ if it reaches the destination through relay $k$. Since, on the average, the probability that a packet serviced from $Q_\ell$ is buffered at the $k$th relay before reaching its destination is $\frac{\lambda_{\ell,k}}{\lambda_\ell}$, the average queueing delays of the primary and secondary packets are given by
     \begin{equation}
 D^{\left(T\right)}_{\rm p}\!=\!D_{\rm p}\!+\!\frac{\sum_{k=1}^N{\lambda_{{\rm p},k}D_{{\rm p},k}}}{\lambda_{\rm p}}, \ \ D^{\left(T\right)}_{\rm s}\!=\!D_{\rm s}\!+\!\frac{\sum_{k=1}^N{\lambda_{{\rm s},k}D_{{\rm s},k}}}{\lambda_{\rm s}}.
 \end{equation}
 A similar approach for computing the end-to-end delay is found in \cite{hamza2013enhanced,huang2013protocol,bao2011stable}.

%%%%%%%%%%%%%%
%%%%%%%%%%%%%%%%%%%%%%%%%%%%%%%%%%%%%
%%%%%%%%%%%%%%%%%%%%%%%
\subsection{Optimization Problems}
\subsubsection{Secondary Throughput Maximization}
Our first optimization problem is concerned with the constrained maximization of the secondary average service rate given $\lambda_{\rm p}$, $\lambda_{\rm s}$ and $N$ subject to predefined tolerable end--to--end mean queuing delay constraints for the primary and secondary packets. Under the ordered acceptance strategy, $\mathcal{S}_{\rm OD}$, the maximum secondary average service rate can be obtained by solving the following problem:
 \begin{equation}\label{opt1}
\begin{aligned}
&\max_{\boldsymbol{\alpha},\boldsymbol{f}_{\rm p},\boldsymbol{f}_{\rm s},\boldsymbol{\omega},\boldsymbol{\rho}^{\left(\rm p\right)},\boldsymbol{\rho}^{\left(\rm s\right)}}  & &  \mu_{\rm s}\\
& \quad\quad\quad {\rm s.t.} & & D^{\left(T\right)}_{\rm s} \le \mathcal{D}^{\left(\mathcal{T}\right)}_{\rm s}, D^{\left(T\right)}_{\rm p} \le \mathcal{D}^{\left(\mathcal{T}\right)}_{\rm p}, \\
& & & \lambda_{\rm s}\!<\! \mu_{\rm s}, \ \lambda_{\rm p}\!<\!\mu_{\rm p}, \ \lambda_{{\rm p},k}\!<\!\mu_{{\rm p},k},\ \lambda_{{\rm s},k}\!<\!\mu_{{\rm s},k}\,\\
& & & 0 \le \boldsymbol{\alpha},\boldsymbol{f}_{\rm p},\boldsymbol{f}_{\rm s} \le 1,\\
& & & 0 \le \boldsymbol{\omega}, \boldsymbol{\rho}^{\left(\rm p\right)},  \boldsymbol{\rho}^{\left(\rm s\right)}, \\
& & & \|\boldsymbol{\omega}\|_1, \|\boldsymbol{\rho}^{\left(\rm p\right)}\|_1, \|\boldsymbol{\rho}^{\left(\rm s\right)}\|_1=1
\end{aligned}
\end{equation}
\noindent where $\mathcal{D}^{\left(\mathcal{T}\right)}_{\rm p} < \infty$ is the maximum tolerable primary end--to--end mean queueing delay, $\mathcal{D}^{\left(\mathcal{T}\right)}_{\rm s} < \infty$ is the maximum tolerable secondary end--to--end mean queueing delay, the notation $a \le \boldsymbol{x}$ is an element wise condition on vector $\boldsymbol{x}$ implying that $a \le x_k$ and $\|\boldsymbol{x}\|_1$ is the $l_1$--norm of the vector $\boldsymbol{x}$ defined as $\|\boldsymbol{x}\|_1=\sum_k |x|_k$. The delay constraints implicitly require the primary, secondary and relays' queues to be stable. The total number of optimization parameters in case of ordered acceptance is $2N!+4N$.

It is worth noting that the optimization problems are solved at a controller which then supplies the required information to the relay stations. The optimal parameters are functions of many parameters such as the channels outage between all nodes in the network (based on the expression in (\ref{choutage}), the channel outage between any two nodes is a function of the packet length, channel bandwidth, SNR, time slot duration, and many other parameters), primary and secondary arrivals rate, delay constraints, number of relays, misdetection probability, and false alarm probability at each relay. Thus, we note that for a given system's parameters, the optimal parameters are fixed as far as these parameters remain constant. Once the optimal parameters are obtained, the controller generates a long sequence of decoding orders and time slot accessing distribution over time slots to be supplied to the relay stations during the whole operational time of the system. This occurs all at once before the actual operation of the system. The optimal acceptance probabilities of users' packets at the relay stations and the probability of selecting one of the relaying queues over the other for a given time slot are all generated locally at each relay station. However, the values of the probabilities are also supplied to the relay stations by the controller all at once before the actual operation of the system.

%The Hessian matrix can be shown to be:
%
%\begin{equation}
%\small \begin{split}
%\begin{array}{cc}
%  &  \\
%   &
%\end{array}
%  \end{split} \normalsize
%\end{equation}

This optimization problem and the others presented in this work are solved numerically Specifically, we use Matlab's fmincon as in \cite{6568963,6177245,4472701,Sult1212:Cooperative,6648968,crowncom,wimob} and the references therein.

Now, we investigate the case in which all relays are set to accept the users' undelivered packets every time slot. Precisely, $f_{{\rm p},k}=f_{{\rm s},k}=1$ for all $k$. Moreover, we assume that the probability of selecting a relaying queue for transmission is $1/(2N)$ where $2N$ is the total number of possibilities. According to the above case, $\mu_{\rm p}=\mu_{\rm p}^{\max}$ in (\ref{pmax1}), $\mu_{\rm s}=\mu_{\rm s}^{\max}$ in (\ref{smax11}), $\pi_{{\rm p},\circ}=1-\lambda_{\rm p}/\mu^{\max}_{\rm p}$, $\pi_{{\rm s},\circ}=1-\lambda_{\rm s}/\mu^{\max}_{\rm s}$,
\begin{equation}
\small \begin{split}
    \lambda_{{\rm p},k}&\!=\!\overline{\pi}_{{\rm p},\circ} P_{\rm p,pd}  \biggr[\!  \overline{P}_{{\rm p},k}    \sum_{\Pi_n}\rho_n^{\left(\rm p\right)}\!\prod_{\substack{{v=1} \\{v \ne k} \\ {m_v<m_k}}}^{N}\!P_{{\rm p},v} \!\biggr],
    \label{opk}
    \end{split} \normalsize
\end{equation}
and
\begin{equation}
\small \begin{split}
\lambda_{{\rm s},k}&\!=\!  \overline{\pi}_{{\rm s},\circ} \pi_{{\rm p},\circ} P_{\rm s,sd}  \biggr[  \overline{P}_{{\rm s},k} \sum_{\Pi_n}\rho_n^{\left(\rm s\right)}\!\prod_{\substack{{v=1} \\{v \ne k} \\ {m_v<m_k}}}^{N} P_{{\rm s},v} \!\biggr].
\label{opk2}
\end{split} \normalsize
\end{equation}
 The relaying queues' mean service rates become constants. That is,
\begin{equation}
\small \begin{split}
    \mu_{{\rm p},k}&\!=\! \frac{1}{2N} \pi_{{\rm p},\circ}\pi_{{\rm s},\circ} \overline{P}_{k,{\rm pd}}, \\ \mu_{{\rm s},k}&\!=\!  \frac{1}{2N} \pi_{{\rm p},\circ} \pi_{{\rm s},\circ} \overline{P}_{k,{\rm sd}}.
    \label{opk3}
    \end{split} \normalsize
\end{equation}
The optimization problem is a convex feasibility problem which can be solved efficiently \cite{boyed}. We note that the objective function is constant. Moreover, $D_{\rm p}$ and $D_{\rm s}$ are constants. We need to prove the convexity of the constraints $D^{\left(T\right)}_{\rm s} \le \mathcal{D}^{\left(\mathcal{T}\right)}_{\rm s}$, $D^{\left(T\right)}_{\rm p} \le \mathcal{D}^{\left(\mathcal{T}\right)}_{\rm p}$, $\lambda_{{\rm p},k}<\mu_{{\rm p},k}$, and $\lambda_{{\rm s},k}<\mu_{{\rm s},k}$. From (\ref{opk}), (\ref{opk2}) and (\ref{opk3}), $\lambda_{{\rm p},k}<\mu_{{\rm p},k}$ and $\lambda_{{\rm s},k}<\mu_{{\rm s},k}$ are linear in $\rho_n^{\left(\rm p\right)}$ and $\rho_n^{\left(\rm s\right)}$.The second term of the queueing delays $D^{\left(T\right)}_{\rm p}$ and $D^{\left(T\right)}_{\rm s}$, $\frac{\sum_{k=1}^N{\lambda_{{\rm p},k}D_{{\rm p},k}}}{\lambda_{\rm p}}$ and $\frac{\sum_{k=1}^N{\lambda_{{\rm s},k}D_{{\rm s},k}}}{\lambda_{\rm s}}$, respectively, are convex if each of the terms inside the summation is convex. Thus, we need to prove the convexity of ${\lambda_{j}D_{{j}}}$, $j\in\{({\rm p},k),({\rm s},k)\}$. The second derivative of $D_j=\lambda_j\frac{1-\lambda_j}{\mu_j-\lambda_j}$ with respect to $\lambda_j$ is given by
\begin{equation}
\small \begin{split}
\frac{\partial^2 D_j}{\partial \lambda_j^2}=2\frac{\mu_j(1-\mu_j)}{(\mu_j-\lambda_j)^3} \ge 0.
  \end{split} \normalsize
\end{equation}
Since the second derivative is positive for $\mu_j\ge \lambda_j$ (stability constraint) and $\mu_j\le1$, the delay constraint is convex over $\lambda_j$.

 For $\mathcal{S}_{\rm RD}$, the maximum secondary average service rate can be obtained by solving the following optimization problem:
 \begin{equation}\label{opt11}
\begin{aligned}
& \max_{\boldsymbol{\alpha},\boldsymbol{f}_{\rm p}, \boldsymbol{f}_{\rm s}, \boldsymbol{\omega}, \boldsymbol{\beta}} & &  \mu_{\rm s} \\
& \quad\quad {\rm s.t.} & & D^{\left(T\right)}_{\rm s} \le \mathcal{D}^{\left(\mathcal{T}\right)}_{\rm s}, D^{\left(T\right)}_{\rm p} \le \mathcal{D}^{\left(\mathcal{T}\right)}_{\rm p}, \\
& & & \lambda_{\rm s}< \mu_{\rm s}, \ \lambda_{\rm p}<\mu_{\rm p}, \ \lambda_{{\rm p},k}<\mu_{{\rm p},k},\ \lambda_{{\rm s},k}<\mu_{{\rm s},k}\,\\
& & & 0 \le \boldsymbol{\alpha},\boldsymbol{f}_{\rm p},\boldsymbol{f}_{\rm s} \le 1,\\
& & & 0 \le \boldsymbol{\omega}, \boldsymbol{\beta}, \\
& & & \|\boldsymbol{\omega}\|_1, \|\boldsymbol{\beta}\|_1=1.
\end{aligned}
\end{equation}
The total number of optimization parameters in the case of random decoding is $5N$.

  For $\mathcal{S}_{\rm RR}$, the maximum secondary average service rate can be obtained by solving an optimization problem similar to (\ref{opt11}) with all elements of $\boldsymbol{\beta}$ equal to $1/N$. The optimization problem is stated as follows:

 \begin{equation}\label{opt11xxx}
\begin{aligned}
& \max_{\boldsymbol{\alpha},\boldsymbol{f}_{\rm p}, \boldsymbol{f}_{\rm s}, \boldsymbol{\omega}} & &  \mu_{\rm s} \\
& \quad\quad {\rm s.t.} & & D^{\left(T\right)}_{\rm s} \le \mathcal{D}^{\left(\mathcal{T}\right)}_{\rm s}, D^{\left(T\right)}_{\rm p} \le \mathcal{D}^{\left(\mathcal{T}\right)}_{\rm p}, \\
& & & \lambda_{\rm s}< \mu_{\rm s}, \ \lambda_{\rm p}<\mu_{\rm p}, \ \lambda_{{\rm p},k}<\mu_{{\rm p},k},\ \lambda_{{\rm s},k}<\mu_{{\rm s},k}\,\\
& & & 0 \le \boldsymbol{\alpha},\boldsymbol{f}_{\rm p},\boldsymbol{f}_{\rm s} \le 1,\\
& & & 0 \le \boldsymbol{\omega},\\
& & & \|\boldsymbol{\omega}\|_1=1.
\end{aligned}
\end{equation}
The total number of optimization variables is equal to $4N$.

Consider the case $f_{{\rm s},k}=f_{k,{\rm p}}=1$. Let $z_k=\alpha_k \omega_k$ and $y_k=(1-\alpha_k) \omega_k$ with $z_k+y_k=\omega_k$. If the queueing delay requirements are large, i.e., $\mathcal{D}^{\left(\mathcal{T}\right)}_{p}=\mathcal{D}^{\left(\mathcal{T}\right)}_{s}=\infty$, which means that the users are delay insensitive, then the optimization problem is a convex feasibility problem. It can be solved as follows:
 \begin{equation}
\begin{aligned}
& \max_{z_k,y_k \forall k} & &  \pi_{{\rm p},\circ} \bigg(\overline{P}_{\rm s,sd}+P_{\rm s,sd} \frac{1}{N}\sum_{k=1}^{N}  \overline{P}_{{\rm s},k}\bigg)={\rm constant}\\
& \quad\quad {\rm s.t.}  & &  \ \lambda_{{\rm p},k}<z_k \pi_{{\rm p},\circ}\pi_{{\rm s},\circ}  \overline{P}_{k,{\rm pd}},\ \lambda_{{\rm s},k}<y_k \pi_{{\rm p},\circ}\pi_{{\rm s},\circ}  \overline{P}_{k,{\rm sd}}\,\\
& & & 0 \le z_k,y_k \le 1,\\
& & & 0 \le z_k,y_k, \\
& & & \sum_{k=1}^N(z_k+y_k)=1
\end{aligned}
\end{equation}
with $\lambda_{\rm s}< \mu_{\rm s}$ and  $\lambda_{\rm p}<\mu_{\rm p}$. The feasible values of $z_k$ and $y_k$ are
 \begin{equation}
\begin{aligned}
\ \frac{\lambda_{{\rm p},k}}{\pi_{{\rm p},\circ}\pi_{{\rm s},\circ}  \overline{P}_{k,{\rm pd}}}\!<\!z_k ,\ \frac{\lambda_{{\rm s},k}}{ \pi_{{\rm p},\circ}\pi_{{\rm s},\circ}  \overline{P}_{k,{\rm sd}}}\!<\!y_k,\ \sum_{k=1}^N(z_k+y_k)\!=\!1
\end{aligned}
\end{equation}

If the users are delay sensitive, the optimization problem can be shown to be a convex feasibility problem. We note that $\lambda_{{\rm p},k}$, $\lambda_{{\rm s},k}$, $\mu_{\rm s}$, and $\mu_{\rm p}$ are constants with respect to the optimization variables, $\omega_k$ and $\alpha_k$. The term $\frac{\sum_{k=1}^N{\lambda_{{\rm s},k}D_{{\rm s},k}}}{\lambda_{\rm s}}$ is convex over $\mu_j$ if each of the terms inside the summation is convex over $\mu_j$. Thus, we need to prove the convexity of ${\lambda_{{\rm s},k}D_{{\rm s},k}}$. The second derivative of $D_j=\lambda_j\frac{1-\lambda_j}{\mu_j-\lambda_j}$ with respect to $\mu_j$ is given by
\begin{equation}
\small \begin{split}
\frac{\partial D_j}{\partial \mu_j}=2\frac{\lambda_j(1-\lambda_j)}{(\mu_j-\lambda_j)^3} \ge 0.
  \end{split} \normalsize
\end{equation}
Since the second derivative is positive, the delay constraint is convex over $\mu_j$. We solve the problem with respect to $z_k$ and $y_k$ then we get the values of $\omega_k$ and $\alpha_k$.

It should be noticed that the total number of optimization parameters is a reflection of both the degrees of freedom and the degree of complexity of the system. Therefore, the ordered acceptance is considered as the strategy with the highest degrees of freedom and the highest complexity among the proposed strategies in this paper. On the other hand, round robin is the simplest strategy among the proposed strategies and it needs less cooperation between the relays than other strategies; it is a cyclic switching operation shared among relays.
\subsubsection{Number of Relays Minimization} Our second formulation is to minimize the number of relays, $N$, needed to achieve certain delay or service rate requirements for the users. Given $\lambda_{\rm p}$ and $\lambda_{\rm s}$ and under the ordered acceptance strategy, $\mathcal{S}_{\rm OD}$, the optimization problem is given by
 \begin{equation}\label{opt2}
\begin{aligned}
&\min_{\boldsymbol{\alpha},\boldsymbol{f}_{\rm p},\boldsymbol{f}_{\rm s},\boldsymbol{\omega},\boldsymbol{\rho}^{\left(\rm p\right)},\boldsymbol{\rho}^{\left(\rm s\right)}}  & & N\\
& \quad\quad\quad {\rm s.t.} & & D^{\left(T\right)}_{\rm s} \le \mathcal{D}^{\left(\mathcal{T}\right)}_{\rm s}, D^{\left(T\right)}_{\rm p} \le \mathcal{D}^{\left(\mathcal{T}\right)}_{\rm p}, \\
& & & \lambda_{\rm s}\!<\! \mu_{\rm s},  \lambda_{\rm p}<\mu_{\rm p},  \lambda_{{\rm p},k}\!<\!\mu_{{\rm p},k}, \lambda_{{\rm s},k}\!<\!\mu_{{\rm s},k}\,\\
& & & 0 \le \boldsymbol{\alpha},\boldsymbol{f}_{\rm p},\boldsymbol{f}_{\rm s} \le 1,\\
& & & 0 \le \boldsymbol{\omega}, \boldsymbol{\rho}^{\left(\rm p\right)},  \boldsymbol{\rho}^{\left(\rm s\right)}, \\
& & & \|\boldsymbol{\omega}\|_1, \|\boldsymbol{\rho}^{\left(\rm p\right)}\|_1, \|\boldsymbol{\rho}^{\left(\rm s\right)}\|_1=1.
\end{aligned}
\end{equation}
In case of $\mathcal{S}_{\rm RD}$, the minimum number of relays required in the network is given by the following optimization problem:
 \begin{equation}\label{opt22}
\begin{aligned}
& \min_{\boldsymbol{\alpha},\boldsymbol{f}_{\rm p}, \boldsymbol{f}_{\rm s}, \boldsymbol{\omega}, \boldsymbol{\beta}} & &  N \\
& \quad\quad {\rm s.t.} & & D^{\left(T\right)}_{\rm s} \le \mathcal{D}^{\left(\mathcal{T}\right)}_{\rm s}, D^{\left(T\right)}_{\rm p} \le \mathcal{D}^{\left(\mathcal{T}\right)}_{\rm p}, \\
& & & \lambda_{\rm s}< \mu_{\rm s}, \ \lambda_{\rm p}<\mu_{\rm p}, \ \lambda_{{\rm p},k}<\mu_{{\rm p},k},\ \lambda_{{\rm s},k}<\mu_{{\rm s},k}\,\\
& & & 0 \le \boldsymbol{\alpha},\boldsymbol{f}_{\rm p},\boldsymbol{f}_{\rm s} \le 1,\\
& & & 0 \le \boldsymbol{\omega}, \boldsymbol{\beta}, \\
& & & \|\boldsymbol{\omega}\|_1, \|\boldsymbol{\beta}\|_1=1.
\end{aligned}
\end{equation}

For $\mathcal{S}_{\rm RR}$, we construct an optimization problem similar to (\ref{opt22}) with all elements in $\boldsymbol{\beta}$ being set to $1/N$.
%Using the aforementioned cases in the first formulation, for large queueing delays, the problems are convex for each $N$. For limited queueing delays case, system $\mathcal{S}_{\rm OD}$, the problems are convex for each $N$ given $(z_k,y_k)$. For system $\mathcal{S}_{\rm RD}$ with limited queueing delays, the problems are convex for each $N$ given $\boldsymbol{\beta}$. Whereas, for system $\mathcal{S}_{\rm RR}$, the problem is convex for each $N$.
%%Using the assumptions that the acceptance factors are set to unity and each relaying queue is assigned with probability $1/(2N)$ for transmission, the optimization problem becomes convex for each $N$.
% We solve a family of convex problems parameterized by $N$. We make grid search over the integer set to find the lowest $N$ such that the users are satisfied.

%%%%%%%%%%%%%%%%%%%%%%%%%%%
%%%%%%%%%%%%%%%%%%%%%%%%%%%
% SENSING ERRORS
%%%%%%%%%%%%%%%%%%%%%%%%%%%
%%%%%%%%%%%%%%%%%%%%%%%%%%%

\section{The Case of Sensing Errors} \label{sensingerror}
We address here the specific scenario of a strong sensing channel between the PU and the SU and consider sensing errors at the relay stations. In other words, we assume that the sensing errors at the SU are negligible, whereas spectrum sensing at the relays may generate erroneous sensing results that should be accounted for. To render the problem tractable and avoid the difficulty of queue interaction due to sensing errors, we impose the assumption that $Q_{{\rm s},k}$ and $Q_{{\rm p},k}$ are never empty. Specifically, when either $Q_{{\rm s},k}$ or $Q_{{\rm p},k}$ is empty, the $k$th relay sends dummy packets.\footnote{The assumption of a node sending dummy packets when it is empty has been considered in many works (see, for example, \cite{el2011opportunistic,sadek,rao1988stability,luo1999stability} and references therein).} The dummy packets do not contribute to the service rates of $Q_{{\rm s},k}$ and $Q_{{\rm p},k}$ but cause interference during concurrent transmission with the primary and secondary terminals. Based on this assumption, the relay scheduled for transmission could cause interference with the primary and secondary transmissions, when it misdetects their transmissions, even if it is empty in the original system. Accordingly, the service rates of the primary and secondary queues are reduced, and the probability of having any of them empty is reduced as well. Consequently, the service rates of the relays are reduced. Therefore, our results provide lower bounds on the primary, secondary and relays service rates.

The $k$th relay scheduled for transmission at a slot misdetects the SU's transmission with probability $P_{{\rm MD}}^{\left(sk\right)}$ and misdetects the PU's transmission with probability $P_{{\rm MD}}^{\left(pk\right)}$. Sensing false alarms have probability $P_{{\rm FA}}^{k}$. All relays are adjusted on the receiving mode and attempt to decode the transmitter packet. The relay scheduled for transmission is the only relay that decides after $2\tau$ seconds relative to the beginning of the time slot about the state of the time slot: busy or free. If the slot is sensed to be free, that relay switches to the transmission mode and start retransmission of one of the packets in its relaying queues. If the channel is sensed to be busy over either interval, the relay continues in the receiving mode. Upon decoding, the relay will be able to identify the packet's origin from the identifier attached to the packet and will use the appropriate decoding order in case of order decoding. In case of random decoding or round robin decoding, one of the relay stations is assigned the decoding task in each time slot. Based on the above, the service rates of the users' queues and the arrival rates of the relaying queues under sensing errors are only affected by the activity of the relay scheduled for transmission. The reduction of the mean service and arrival rates is equal to $\sum_r \omega_r \mathcal{B}^{\left(\ell r\right)}$, where $\mathcal{B}^{\left(\ell r\right)}$ denotes the complement of the probability that relay $r$, scheduled for transmission, erroneously finds the time slot free given there is an active transmission from user $\ell$.

%%%%%%%%%%%%%%
%%%%%%%%%%%%%%%%%%%%%%%%%%%%%%%%%%%%%
%%%%%%%%%%%%%%%%%%%%%%%

%\subsection{Ordered Acceptance}
Now we compute $\mathcal{B}^{\left(\ell r\right)}$ for both users. The relay $r$ scheduled for transmission disrupts the primary if it fails to detect the activity of the PU during both sensing intervals. That is, the probability that the $r$th relay detects the time slot as a busy slot due to activity of PU is $\mathcal{B}^{\left(pr\right)}=(1-P_{\rm MD}^{\left(pr\right)} P_{\rm MD}^{\left(pr\right)})$.\footnote{As mentioned in Section \ref{sec2}, we assume in this paper that if two terminals transmit simultaneously, their packets cannot be decoded correctly at the respective receivers.} The probability that the relay $r$ scheduled for transmission does not disrupt the secondary activity is equal to the probability that the relay either detects the secondary transmission, or falsely finds the PU to be active while it is not. In either case, it will abstain from transmission, thereby avoiding collision with the secondary transmission. Thus, the probability is given by $\mathcal{B}^{\left(sr\right)} = \left(1-P_{\rm MD}^{\left(sr\right)}\big[1-P_{\rm FA}^{r}\big]\right)$. Accordingly, we have the following set of arrival and service rates:

\begin{equation}
\small \begin{split}
\label{xxxy}
&\mu^{\left({\rm SE}\right)}_\ell\!=\! \mu_\ell \ \sum_{r=1}^{N}\omega_r\ \mathcal{B}^{\left(\ell r\right)}, \ \ell\in\{{\rm p},s\}
  \end{split} \normalsize
\end{equation}

\begin{equation}
\small \begin{split}
    &\lambda^{\left({\rm SE}\right)}_{\ell,k}=   \lambda_{\ell,k} \sum_{r=1}^{N}\omega_r\ \mathcal{B}^{\left(\ell r\right)}, \ \ell\in\{{\rm p},s\}
    \end{split} \normalsize
\end{equation}

\begin{equation}
\small \begin{split}
    \mu^{\left({\rm SE}\right)}_{\ell,k}&\!=\! \mu_{\ell,k} \big(1-P_{\rm FA}^{k}\big)^2, \ \ell\in\{{\rm p},s\}
    %,\ \\ \mu^{\left({\rm SE}\right)}_{{\rm s},k}\!&=\! \mu_{{\rm s},k}  \big(1-P_{\rm FA}^{k}\big)^2
    \label{df}
    \end{split} \normalsize
\end{equation}
 where $\mu^{\left({\rm SE}\right)}_\ell$, $\lambda^{\left({\rm SE}\right)}_{\ell,k}$ and $ \mu^{\left({\rm SE}\right)}_{\ell,k}$ are the rates of users and relays in the case of sensing error and $\mu_\ell$, $\lambda_{\ell,k}$ and $ \mu_{\ell,k}$ are the rates of users and relays in the case of perfect sensing.\footnote{These values depend on the decoding strategy used as explained earlier.} We note that the term $\big(1-P_{\rm FA}^{k}\big)^2$ in (\ref{df}) represents the probability that the $k$th relay finds the time slot free from transmissions. This equals the probability that the sensor of relay $k$ does not generate false alarm over both sensing intervals.

 To obtain the optimal secondary average service rate in case of sensing errors, we construct an optimization problem similar to (\ref{opt1}) and (\ref{opt11}). For the minimum number of relays needed to achieve certain QoS constraints, we construct an optimization problem similar to (\ref{opt2}) and (\ref{opt22}).

%%%%%%%%%%%%%%%%%%%%%%%%%%%
%%%%%%%%%%%%%%%%%%%%%%%
%%%%%%%%%%%%%%%%%%%%%%%%%%%%%%%
%%%%%%%%%%%%%%%%%%%%%%%%%%%%%%%%%
% NUMERICAL RESULTS
%%%%%%%%%%%%%%%%%%%%%%%%%%%%%%%%
%%%%%%%%%%%%%%%%%%%%%%%%%%%%%%%
%%%%%%%%%%%%%%%%%%%%%%%%%%%%%%%%
\section{Numerical Results}\label{num}
\begin{table}
\caption{Relays channels outage probabilities where $N=5$ and $\tau_{\rm f}= 0$.}
\begin{center}
\begin{tabular}{  |@{}c@{} |c |c|c |@{}c@{} }
    \hline\hline
   \hbox{Relay-SD}& \hbox{Relay-PD} & \hbox{SU-relay}& \hbox{PU-relay} \\[5pt]\hline    $P_{1,{\rm SD}}=0.1$ &$P_{1,{\rm PD}}=0.1$&$P_{{\rm s},1}=0.1$ &$P_{{\rm p},1}=0.1$\\ $P_{2,{\rm SD}}=0.1$ &$P_{2,{\rm PD}}=0.1$ &$P_{{\rm s},2}=0.1$& $P_{{\rm p},2}=0.02$\\  $P_{3,{\rm SD}}=0.2$ &$P_{3,{\rm PD}}=0.2$ &$P_{{\rm s},3}=0.02$ & $P_{{\rm p},3}=0.2$\\  $P_{4,{\rm SD}}=0.1$ &$P_{4,{\rm PD}}=0.01$&$P_{{\rm s},4}=0.1$ &$P_{{\rm p},4}=0.1$\\ $P_{5,{\rm SD}}=0.01$&$P_{5,{\rm PD}}=0.01$ &$P_{{\rm s},5}=0.01$& $P_{{\rm p},5}=0.01$ \\[5pt]\hline
\end{tabular}

\label{table1}
\end{center}
\end{table}
\begin{table*}
\caption{Channel parameters for three relays.}
\begin{center}
\begin{tabular}{ |c |c | c|c|c|c|c|c|}
    \hline\hline
           \hbox{ Parameter} & \hbox{Value}& \hbox{ Parameter} & \hbox{Value}& \hbox{ Parameter} & \hbox{Value}& \hbox{ Parameter} & \hbox{Value} \\[5pt]\hline
          % $N$& 3\\[5pt]\hline
            $b$&1000 bits&
            $W$ &$10$ MHz &
            $T$ & $10^{-3}$ seconds &
            $\sigma_{1,SD}$ & 0.8 \\[5pt]\hline
            $\sigma_{2,SD}$& 0.75 &
            $\sigma_{3,SD}$ &0.9 &
$\sigma_{1,PD}$&0.88 &
$\sigma_{2,PD}$&0.95 \\[5pt]\hline
$\sigma_{3,PD}$&0.85 &
$\sigma_{SU,1}$&0.83 &
$\sigma_{SU,2}$&0.92 &
$\sigma_{SU,3}$&0.79  \\[5pt]\hline
$\sigma_{PU,1}$&0.82&
$\sigma_{PU,2}$&0.935 &
$\sigma_{PU,3}$&0.815 & % variance of the fading channel
$\tau$& 0.1 T \\[5pt]\hline% Sensing duration
$\gamma_{{\rm s},sd}$&2 & %Signal to noise ratio if the user transmits from the beginning of the time slot
$\gamma_{{\rm p},pd}$&3&
$\gamma_{PU,1}$&3 &
$\gamma_{PU,2}$&2.5 \\[5pt]\hline
$\gamma_{PU,3}$&2 &
    $\gamma_{SU,1}$&3 &
    $\gamma_{SU,2}$&2.5 &
$    \gamma_{SU,3}$&2 \\[5pt]\hline
$\sigma_{{\rm p},pd}$&1 & & & & & &  \\[5pt]\hline
%$\sigma_{{\rm s},sd}$&$0.8$ \\[5pt]\hline
\end{tabular}
\label{table2}
\end{center}
\end{table*}

In this section, we provide some numerical results for the optimization problems considered in this paper. Figs. \ref{rr10} and \ref{r11} demonstrate the case of negligible feedback duration, i.e., $\tau_{\rm f}=0$, and low outage probabilities for the PU-PD and the SU-SD direct links: $P_{\rm s,sd}=0.2$ and $P_{\rm p,pd}=0.1$. The figures are generated using $N=2$, $\lambda_{\rm s}=0.1$ packets per time slot, $D^{\left(T\right)}_{\rm p}\le 1.6$ time slots, $D^{\left(T\right)}_{\rm s}\le 3$ time slots, $\lambda_{\rm s}=0.2$ packets per time slot, and the outage probabilities given in the first two lines of Table \ref{table1}. As evident from Fig. \ref{rr10}, the ordered acceptance strategy with two relays almost achieves the upper bound on the secondary average service rate, which is equal to $1-\lambda_{\rm p}$. Random assignment and round robin decoding give almost the same performance for the parameters used in the simulation. The primary average service rate, as shown in Fig. \ref{r11}, is constant and almost unity for the proposed decoding strategies compared to $1-P_{\rm p,pd}=0.9$ when no relays are used. The primary mean service rate is constant because the solution of the optimization problem makes $\mu_{\rm p}=\mu_{\rm p}^{\rm \max}$ (see expressions (\ref{pmax1}), (\ref{pmax2}), and (\ref{pmax3})).

\begin{figure}
\centering
  \includegraphics[width=1\columnwidth]{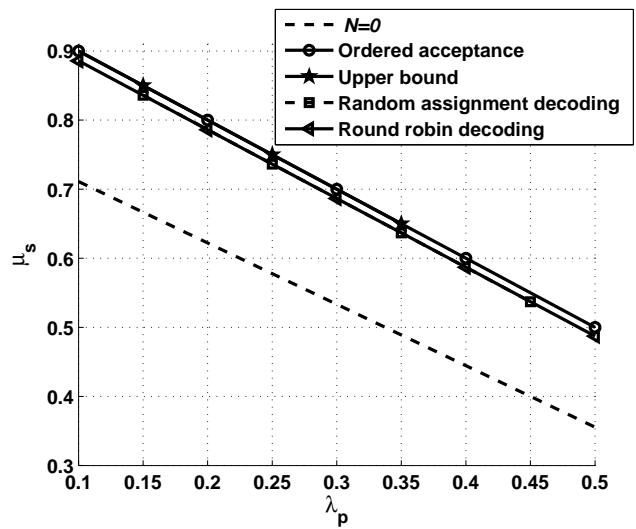}\\
  \caption{Optimal secondary average service rate versus the primary average arrival rate, $\lambda_{\rm p}$.}
\label{rr10}
\end{figure}

\begin{figure}
\centering
  \includegraphics[width=1\columnwidth]{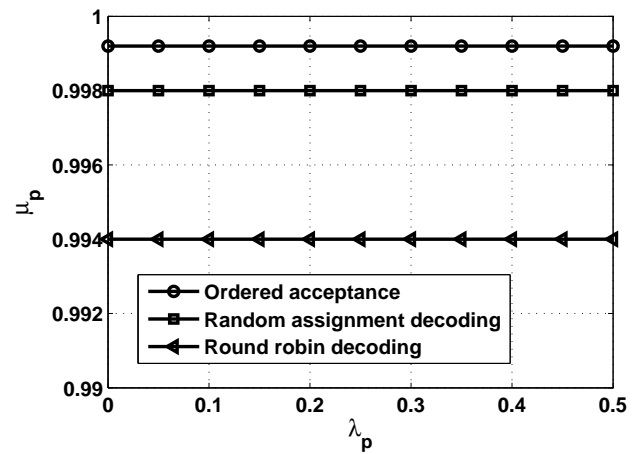}\\
  \caption{Primary average service rate versus $\lambda_{\rm p}$ for the same parameters used to generate Fig. \ref{rr10}. }\label{r11}
\end{figure}

Fig. \ref{rr100} reveals the impact of increasing the number of relays on the optimal secondary average service rate for $\mathcal{S}_{\rm RD}$ with $\tau_{\rm f}\!=\!0$. This figure is generated using $P_{\rm s,sd}=0.3$ and $P_{\rm p,pd}=0.4$, $\lambda_{\rm s}=0.2$ packets per time slot, $D^{\left(T\right)}_{\rm p}\le 5$ time slots, $D^{\left(T\right)}_{\rm s}\le 10$ time slots, and the outage probabilities given in Table \ref{table1}. As shown in the figure, when the number of relays, $N$, increases, the average service rate of the SU (maximum $\mu_{\rm s}$) approaches the upper bound.

Figs. \ref{r12} and \ref{r13} also show the case of $\tau_{\rm f}=0$, but this time there are no direct links between the PU and the SU and their respective receivers. That is, $P_{\rm s,sd}=1$ and $P_{\rm p,pd}=1$. The parameters used to generate these figures are: $\tau_{\rm f}=0$, $N=2$, $\lambda_{\rm s}=0.1$ packets per time slot, $D^{\left(T\right)}_{\rm p}\le 25$, $D^{\left(T\right)}_{\rm s}\le 25$, and the outage probabilities given in the first three lines of Table \ref{table1}. Note that in this case relaying is essential since without cooperation (no relays), the primary service rate is equal to $1-P_{{\rm p},{\rm PD}}=0$ and both the primary and secondary queues are always backlogged and unstable and packets are never being served. Hence, the queueing delay of each user is infinity.

\begin{figure}
\centering
  \includegraphics[width=1\columnwidth]{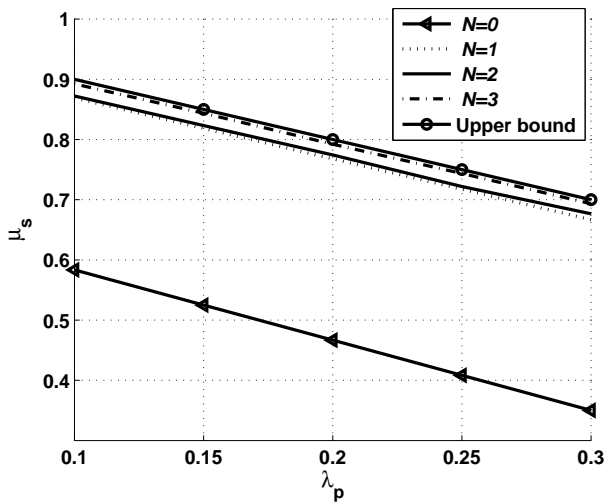}\\
  \caption{Impact of increasing the number of relays on the optimal secondary average service rate for $\mathcal{S}_{\rm RD}$.}
  \label{rr100}
\end{figure}

\begin{figure}
\centering
  \includegraphics[width=1\columnwidth]{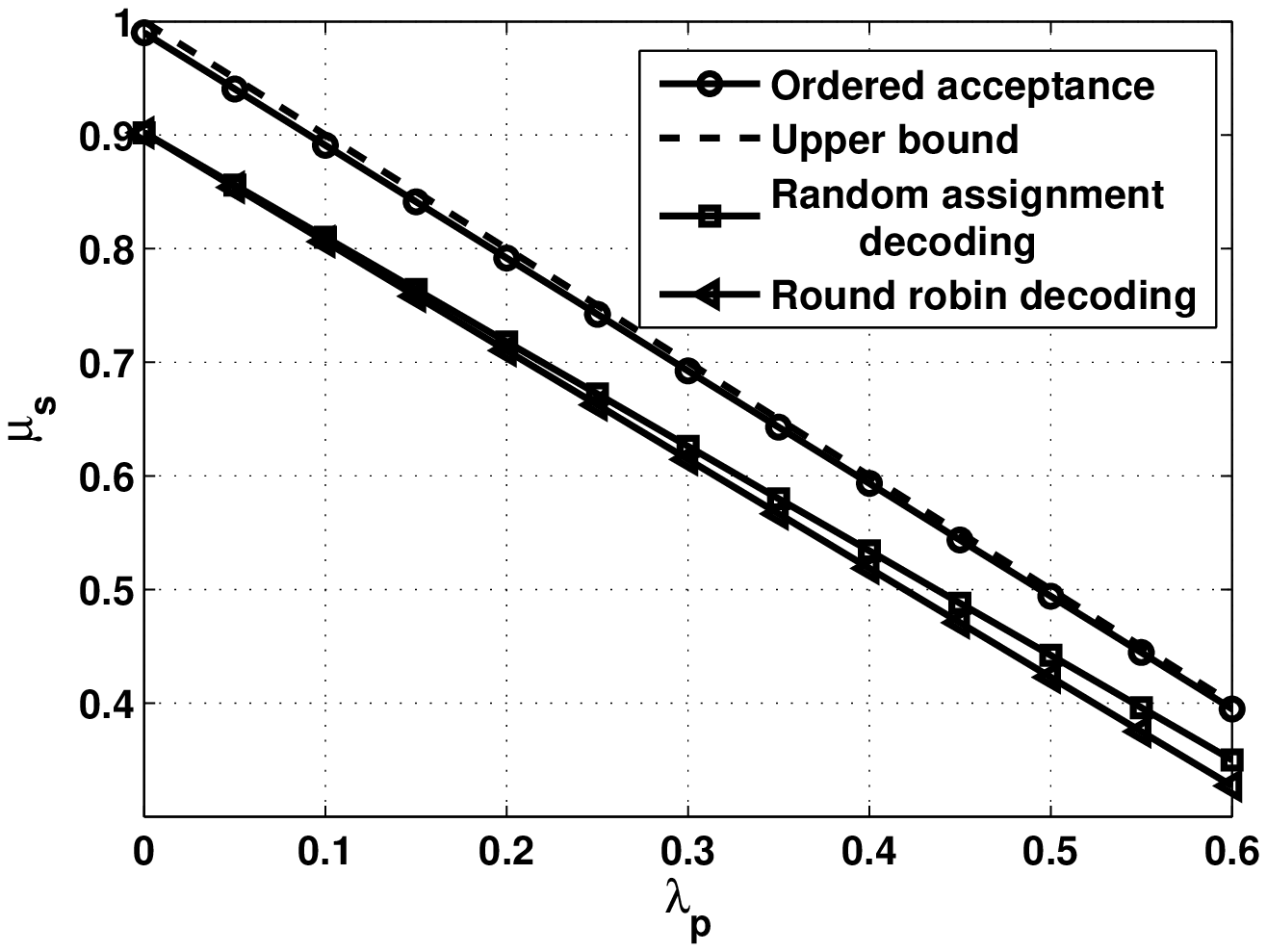}\\
  \caption{Optimal secondary average service rate versus $\lambda_{\rm p}$ for the case of no direct links between the PU and the SU and their respective receivers.}\label{r12}
\end{figure}

\begin{figure}
\centering
  \includegraphics[width=1\columnwidth]{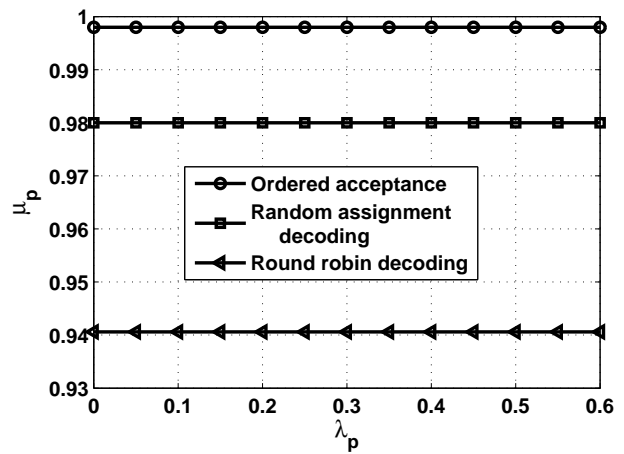}\\
  \caption{Primary average service rate versus $\lambda_{\rm p}$ for the case of no direct links between the PU and the SU and their respective receivers.}\label{r13}
\end{figure}

Fig. \ref{r15} represents the solution of the optimization problems (\ref{opt2}) and (\ref{opt22}), which is the number of relays required to achieve specific QoS requirements for the PU and the SU. The parameters used to generate the figure are: $P_{\rm s,sd}=0.8$, $P_{\rm p,pd}=0.2$, $\lambda_{\rm s}=0.1$ packets per time slot, and channel outage probabilities provided in Table \ref{table1}.

\begin{figure}
\centering
  \includegraphics[width=1\columnwidth]{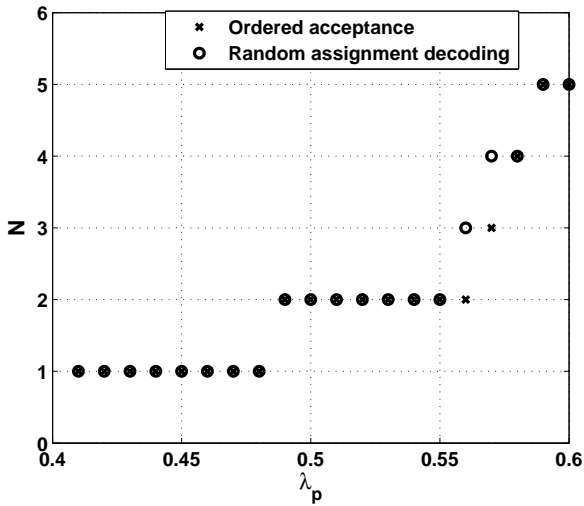}\\
  \caption{Number of relays required to achieve the QoS requirements: $D^{\left(T\right)}_{\rm p}\le 2$ time slots, and $D^{\left(T\right)}_{\rm s}\le 15$ time slots.}\label{r15}
 % $R_{\rm p}\ge \lambda_{\rm p}$, and $R_{\rm s}\ge 0.1$
\end{figure}

Fig. \ref{r16} shows the impact of feedback duration on the maximum SU's average service rate. As $\tau_{\rm f}$ increases, the ordered acceptance strategy loses its edge and is outperformed by the random assignment strategy. The figure is generated using $N=2$, $\lambda_{\rm s}=0.4$ packets per time slot, $D^{\left(T\right)}_{\rm p}\le 5$ time slots, $D^{\left(T\right)}_{\rm s}\le 5$ time slots, the channel parameters of relay $1$ and relay $2$ provided in Table \ref{table2}, and $\sigma_{{\rm s},sd}=0.4$.

\begin{figure}
\centering
  \includegraphics[width=1\columnwidth]{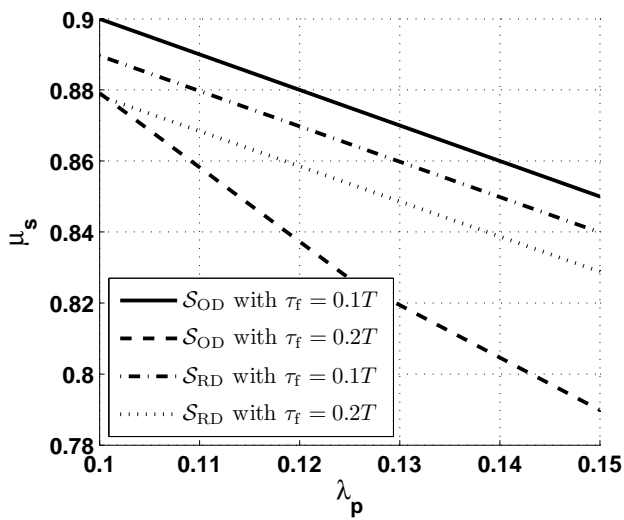}\\
  \caption{Impact of non-zero feedback duration on the maximum secondary average service rate.}\label{r16}
\end{figure}

 Fig. \ref{r17} represents the secondary average service rate for the case $N=3$,  $\sigma_{{\rm s},sd}=0.8$, and $\tau_{\rm f}=0.05T$ in the presence of sensing errors at the relays' spectrum sensors. The figure is for parameters $N=3$, $\lambda_{\rm s}=0.2$ packets per time slot, $D^{\left(T\right)}_{\rm p}\le 3$ time slots, and $D^{\left(T\right)}_{\rm s}\le 120$ time slots. The sensing error probabilities are: $P_{\rm MD}^{\left(p1\right)}=0.1$, $P_{\rm MD}^{\left(p2\right)}=0.09$, $P_{\rm MD}^{\left(p3\right)}=0.12$, $P_{\rm MD}^{\left(s1\right)}=0.1$, $P_{\rm MD}^{\left(s2\right)}=0.068$, $P_{\rm MD}^{\left(s3\right)}=0.09$, $P_{\rm FA}^{1}=0.05$, $P_{\rm FA}^{2}=0.04$, and $P_{\rm FA}^{3}=0.03$. The probabilities of correct reception over the channels between the sources and relays, and the relays and destinations can be computed using the parameters in Table \ref{table2} and expression (\ref{correctreception}). Note that because $\tau_{\rm f}$ is nonzero, the outage probabilities differ for the different strategies due to the difference in the feedback duration $T_{\rm F}$ as explained in Section \ref{sec2}.

\begin{figure}
\centering
  \includegraphics[width=1\columnwidth]{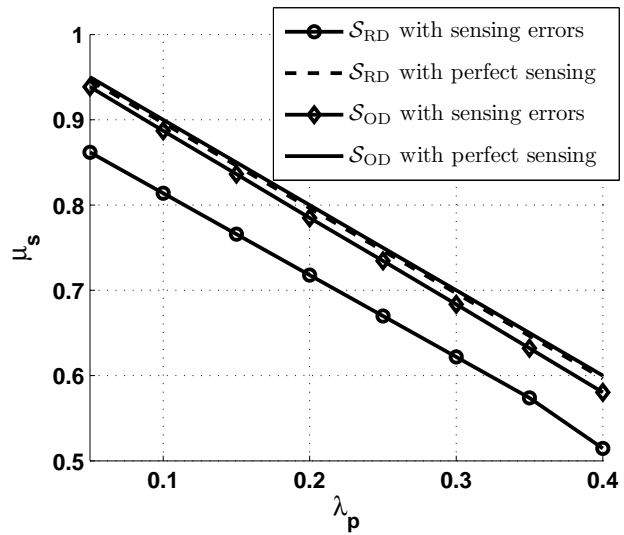}\\
  \caption{Impact of spectrum sensing errors on the maximum secondary average service rate under relatively small feedback duration $\tau_{\rm f} = 0.05T$.}\label{r17}
\end{figure}

Fig. \ref{r1500} shows the mean service rate of the SU in the case of sensing errors and considerable feedback duration per relay. The parameters used to generate the figure are exactly those of Fig. \ref{r17} with $N=2$ and $\tau_{\rm f}=0.24T$. It is noted that $\mathcal{S}_{\rm RD}$ outperforms $\mathcal{S}_{\rm OD}$ in the case of perfect sensing and sensing errors. This is because of the high transmission time losses due to the time consumed in channel feedback coordination in the case of ordered acceptance. In particular, for $\mathcal{S}_{\rm OD}$, the overall feedback duration is $T_{\rm F}=(N+1)\tau_{\rm f}=3\times 0.24T=0.72T$, whereas for $\mathcal{S}_{\rm RD}$, $T_{\rm F}=2\tau_{\rm f}=0.48T$.

Fig. \ref{r19} investigates the minimum number of relays in the case of ordered acceptance with and without sensing errors. The parameters used to generate this figure are the same as those of Fig. \ref{r17}. As is evident from the figure, spectrum sensing errors may cause an increase in the minimum number of relays required to satisfy the primary and secondary queueing delay constraints.

\begin{figure}
\centering
  \includegraphics[width=1\columnwidth]{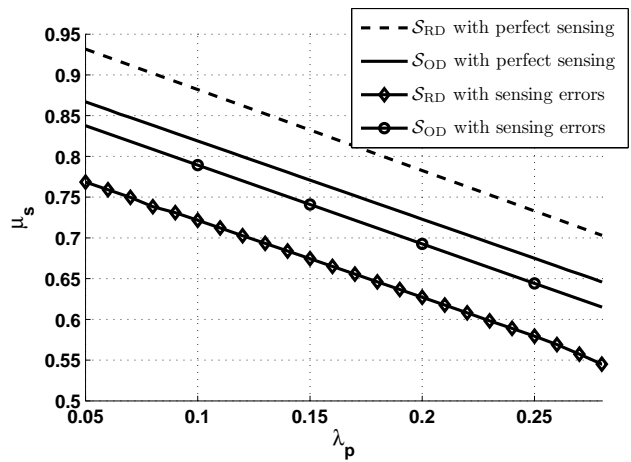}\\
  \caption{Impact of spectrum sensing errors on the maximum secondary average service rate under relatively large feedback duration $\tau_{\rm f} = 0.24T$.}\label{r1500}
  %$R_{\rm p}\ge \lambda_{\rm p}$, and $R_{\rm s}\ge 0.13$
\end{figure}

\begin{figure}
\centering
  \includegraphics[width=1\columnwidth]{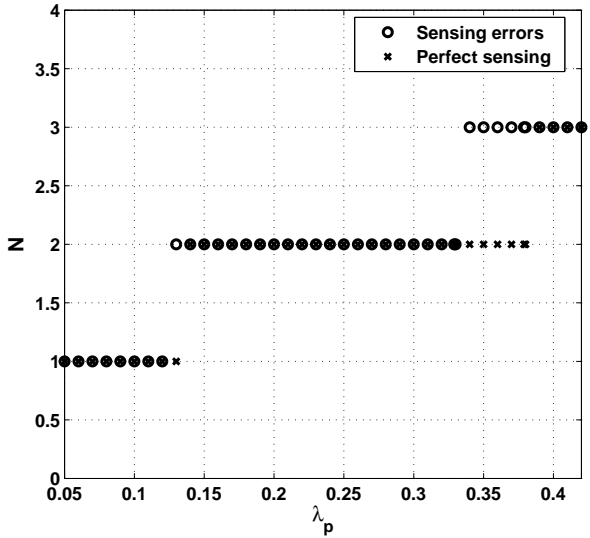}\\
  \caption{Minimum number of relays in the case of ordered acceptance with and without sensing errors. The feedback process duration per relay $\tau_{\rm f}=0.05T$, $D^{\left(T\right)}_{\rm p}\le 10$ time slots, and $D^{\left(T\right)}_{\rm s}\le 20$ time slots.}\label{r19}
  %$R_{\rm p}\ge \lambda_{\rm p}$, and $R_{\rm s}\ge 0.13$
\end{figure}

\section{Conclusion} \label{conc}
In this paper, we have investigated the use of multiple relays to satisfy pre--specified queuing delay constraints on the primary and secondary transmissions. We have proposed and investigated three relay decoding strategies; and have seen that the ordered acceptance strategy maintains the best performance under negligible feedback duration. Our work has assumed knowledge of channel statistics, but not the instantaneous values of the channel gains. Two interesting extensions of this work would be the incorporation of the knowledge of the instantaneous values of the channels, in addition to the queue state information, into the relay scheduling decisions and the investigation of the possibility of cooperation among the relays by forming a virtual antenna array.

% or


\appendices
\section{Proof of Theorem 1}
 In this section, we prove the advantage of strategy $\mathcal{S}_{\rm OD}$ over $\mathcal{S}_{\rm RD}$ and $\mathcal{S}_{\rm RR}$ for a negligible feedback duration per relay, i.e., $\tau_{\rm f}\approx \kappa \ll T $ and $\kappa\!\rightarrow\!0$. We first focus on the perfect sensing case assuming that $\tau$ is large enough to render negligible the probabilities of misdetection and false alarm; then we prove the case of sensing errors. We compare the nodes' service rates of the queues in the proposed strategies with each other.
\subsection{The Case of Perfect Sensing}
\begin{proof}
 For strategy $\mathcal{S}_{\rm OD}$, we define $\epsilon^{\left(\ell\right)}_{m_kk}$ as the probability of assigning the $m_k$th decoding rank to the $k$th relay. If the received packet comes from the PU, $\ell=p$, whereas if the received packet comes from the SU, $\ell=s$. The summation over these probabilities satisfies the constraints
\begin{equation}
\small \begin{split}
\sum_{m_k=1}^{N}\epsilon^{\left(\ell\right)}_{m_kk}&= 1, \forall k\in\{1,2,\dots,N\}, \\ \sum_{k=1}^{N}\epsilon^{\left(\ell\right)}_{m_kk}&= 1, \forall m_k\in\{1,2,\dots,N\}
\end{split} \normalsize
\end{equation}
 where $\ell\in\{{\rm p},s\}$. It should be noted that $\epsilon^{\left(\ell\right)}_{{m_k}k}$, the probability that rank $m_k$ is assigned to relay $k$, relates to the $q^{\left(\ell\right)}$'s as follows:
\begin{equation}
\small \begin{split}
\epsilon^{\left(\ell\right)}_{{m_k}k}=\sum_{{\sim}{m_k}}q^{\left(\ell\right)}(m_1,m_2,\dots,m_{N}), \forall k,m_k\in\{1,2,\dots,N\}
\end{split} \normalsize
\label{constraints_omega_1}
\end{equation}
\noindent where the sum is over all indices except $m_k$. Hereinafter, we add superscripts to the mean service rates to indicate the strategies to which those rates belong.

The mean service rate of the primary queue in strategy $\mathcal{S}_{\rm OD}$ is
%\begin{equation}
%\small \begin{split}
%  \mu^{\left(\mathcal{S}_{\rm OD}\right)}_{\rm p}&\!=\!\overline{P}_{\rm p,pd}+P_{\rm p,pd}\sum_{k\!=1\!}^{N} \biggr[  \overline{P}_{{\rm p},k} f_{{\rm p},k}  \underset{(m_1,m_2,\dots,m_{N})}{\sum}\rho^{\left(\rm p\right)}(m_1,m_2,\dots,m_N) \prod_{\substack{{v=1} \\{v \ne k} \\ {m_v<m_k}}}^{N}\overline{\overline{P}_{{\rm p},v} f_{{\rm p},v}}\biggr]\\& \!=\!\overline{P}_{\rm p,pd}\!+\!P_{\rm p,pd}\sum_{k=1}^{N}   \overline{P}_{{\rm p},k} f_{{\rm p},k}  \bigg[\epsilon^{\left(\rm p\right)}_{1k}\!+\!\sum_{\substack{{(m_1,m_2,\dots,m_{N})}\\{m_k\ne 1}}}\rho^{\left(\rm p\right)}(m_1,m_2,\dots,m_N) \! \prod_{\substack{{v=1} \\{v \ne k} \\ {m_v<m_k}}}^{N}\!\overline{\overline{P}_{{\rm p},v} f_{{\rm p},v}}\bigg]\\&=\overline{P}_{\rm p,pd}\!+\!P_{\rm p,pd}\sum_{k=1}^{N}   \overline{P}_{{\rm p},k} f_{{\rm p},k}  \epsilon^{\left(\rm p\right)}_{1k}\!\\&\,\,\,\,\,\,\,\,\,\,\,\,\,\,\ +P_{\rm p,pd}\sum_{k=1}^{N}   \overline{P}_{{\rm p},k} f_{{\rm p},k}  \!\sum_{\substack{{(m_1,m_2,\dots,m_{N})}\\{m_k\ne 1}}}\rho^{\left(\rm p\right)}(m_1,m_2,\dots,m_N) \! \prod_{\substack{{v=1} \\{v \ne k} \\ {m_v<m_k}}}^{N}\!\overline{\overline{P}_{{\rm p},v} f_{{\rm p},v}}
%  \label{eqn1}
%  \end{split} \normalsize
%\end{equation}
\begin{equation}
\small \begin{split}
  \mu^{\left(\mathcal{S}_{\rm OD}\right)}_{\rm p}&\!=\!\overline{P}_{\rm p,pd}\!+\!P_{\rm p,pd}\sum_{k\!=1\!}^{N} \Biggr[\!\overline{P}_{{\rm p},k} f_{{\rm p},k}\underset{\Pi_n}{\sum}\rho^{\left(\rm p\right)}_n \prod_{\substack{{v\!=\!1} \\{v \!\ne\! k} \\ {m_v<m_k}}}^{N}\!\overline{\overline{P}_{{\rm p},v} f_{{\rm p},v}}\!\Biggr].
  \label{eqn1}
  \end{split} \normalsize
\end{equation}
Using (\ref{constraints_omega_1}) and noting that $\prod_{\substack{{v=1} \\{v \ne k} \\ {m_v<1}}}^{N}\!\overline{\overline{P}_{{\rm p},v} f_{{\rm p},v}}=1$, we have
\begin{equation}
\small \begin{split}
  &\underset{\Pi_n}{\sum}\rho^{\left(\rm p\right)}_n \prod_{\substack{{v=1} \\{v \ne k} \\ {m_v<m_k}}}^{N}\overline{\overline{P}_{{\rm p},v} f_{{\rm p},v}}\!\\& \!=\!\epsilon^{\left(\rm p\right)}_{1k}\! \!+\!\sum_{\substack{{\Pi_n}\\{m_k\ne 1}}}\Biggr[\rho^{\left(\rm p\right)}_n \prod_{\substack{{v=1} \\{v \ne k} \\ {m_v<m_k}}}^{N}\!\overline{\overline{P}_{{\rm p},v} f_{{\rm p},v}}\Biggr]
  \label{xxxx}
  \end{split} \normalsize
\end{equation}
Substituting (\ref{xxxx}) into (\ref{eqn1}), we get
\begin{equation}
\small \begin{split}
  \mu^{\left(\mathcal{S}_{\rm OD}\right)}_{\rm p}&\!=\!\overline{P}_{\rm p,pd}\!+\!P_{\rm p,pd}\sum_{k=1}^{N}   \overline{P}_{{\rm p},k} f_{{\rm p},k}  \big(\epsilon^{\left(\rm p\right)}_{1k}\!+\!\zeta_{\rm p}\!\big)
%  \\&=\overline{P}_{\rm p,pd}\!+\!P_{\rm p,pd}\sum_{k=1}^{N}   \overline{P}_{{\rm p},k} f_{{\rm p},k}  \epsilon^{\left(\rm p\right)}_{1k}\! +\!P_{\rm p,pd}\sum_{k=1}^{N}   \overline{P}_{{\rm p},k} f_{{\rm p},k}  \zeta_{\rm p}
  \end{split} \normalsize
\end{equation}
where
\begin{equation}
\small \begin{split}
 \zeta_{\rm p}\!=\!\!\sum_{\substack{{(m_1,m_2,\dots,m_{N})}\\{m_k\ne 1}}}\rho^{\left(\rm p\right)}(\!m_1,m_2,\dots,m_N\!) \! \prod_{\substack{{v\!=\!1} \\{v \!\ne\! k} \\ {m_v\!<\!m_k}}}^{N}\!\overline{\overline{P}_{{\rm p},v} f_{{\rm p},v}}\!\ge\!0.
  \end{split} \normalsize
\end{equation}
 Recall that for strategy $\mathcal{S}_{\rm RD}$, the primary and secondary mean service rates are given by
\begin{equation}
\small \begin{split}
  \mu^{\left(\mathcal{S}_{\rm RD}\right)}_{\rm p}& \!=\!\overline{P}_{\rm p,pd}+P_{\rm p,pd}\sum_{k\!=\!1}^{N}   \overline{P}_{{\rm p},k} f_{{\rm p},k} \beta_k, \\
    \mu^{\left(\mathcal{S}_{\rm RD}\right)}_{\rm s}\!&=\! \pi^{\left(\mathcal{S}_{\rm RD}\right)}_{{\rm p},\circ} \bigg(\!\overline{P}_{\rm s,sd}+P_{\rm s,sd} \sum_{k=1}^{N}  \overline{P}_{{\rm s},k} f_{{\rm s},k}\beta_k\!\bigg).
  \end{split} \normalsize
\end{equation}
%Setting $\epsilon^{\left(\rm p\right)}_{1k}=\beta_k$
%Since $0\le \epsilon^{\left(\rm p\right)}_{1k},  \epsilon^{\left(\rm s\right)}_{1k}, \beta_k\le 1$, $\sum_{k=1}^{N} \beta_k=1$ and $\sum_{k=1}^{N} \epsilon^{\left(\ell\right)}_{1k}=1$, then if we set $ \epsilon^{\left(\rm p\right)}_{1k}=\epsilon^{\left(\rm s\right)}_{1k}= \beta_k$ and substitute in (\ref{eqn1}), we get
Subtracting $\mu^{\left(\mathcal{S}_{\rm RD}\right)}_{\rm p}$ from $\mu^{\left(\mathcal{S}_{\rm OD}\right)}_{\rm p}$, we obtain
\begin{equation}
\small \begin{split}
  \mu^{\left(\mathcal{S}_{\rm OD}\right)}_{\rm p}-\mu^{\left(\mathcal{S}_{\rm RD}\right)}_{\rm p}&\!=P_{\rm p,pd}\sum_{k=1}^{N}   \overline{P}_{{\rm p},k} f_{{\rm p},k}  \big(\epsilon_{1k}-\beta_k\big)\!\\&\,\,\,\,\,\,\,\,\,\,\,\,\,\,\ +P_{\rm p,pd}\sum_{k=1}^{N}   \overline{P}_{{\rm p},k} f_{{\rm p},k}  \zeta_{\rm p}.
  \label{eqn4}
  \end{split} \normalsize
\end{equation}
Since $0\le \epsilon^{\left(\rm p\right)}_{1k},  \epsilon^{\left(\rm s\right)}_{1k}, \beta_k\le 1$, $\sum_{k=1}^{N} \beta_k=1$ and $\sum_{k=1}^{N} \epsilon^{\left(\ell\right)}_{1k}=1$, therefore, we can set $ \epsilon^{\left(\rm p\right)}_{1k}\!=\!\epsilon^{\left(\rm s\right)}_{1k}\!=\! \beta_k$. Accordingly, the first term on the right-hand side (RHS) of (\ref{eqn4}) is equal to zero, and we have
\begin{equation}
\small \begin{split}
  \mu^{\left(\mathcal{S}_{\rm OD}\right)}_{\rm p}&\!=P_{\rm p,pd}\sum_{k=1}^{N}   \overline{P}_{{\rm p},k} f_{{\rm p},k}  \zeta_{\rm p}\ge\mu^{\left(\mathcal{S}_{\rm RD}\right)}_{\rm p}.
  \label{eqn3}
  \end{split} \normalsize
\end{equation}
The probability that a queue, $Q_\ell$, belonging to a system operating under strategy $\mathcal{S}_{\rm \mathcal{J}}$ is empty is given by $\pi^{\left(\mathcal{S}_{\rm \mathcal{J}}\right)}_{\ell,\circ}
\!=\!1\!-\!\frac{\lambda_\ell}{\mu^{\left(\mathcal{S}_{\rm \mathcal{J}}\right)}_\ell}$  where ${\rm \mathcal{J}\in\{ OD,RD\}}$ and $\ell\in\{{\rm p},s\}$. Since $\mu^{\left(\mathcal{S}_{\rm OD}\right)}_{\rm p}\ge \mu^{\left(\mathcal{S}_{\rm RD}\right)}_{\rm p}$, therefore, $\pi^{\left(\mathcal{S}_{\rm OD}\right)}_{{\rm p},\circ}\!\ge\! \pi^{\left(\mathcal{S}_{\rm RD}\right)}_{{\rm p},\circ}$. In a similar fashion, the mean service rate of the secondary queue can be lower bounded as
\begin{equation}
\small \begin{split}
  \mu^{\left(\mathcal{S}_{\rm OD}\right)}_{\rm s}\!&\ge\!\bigg(\overline{P}_{\rm s,sd}\!+\!P_{\rm s,sd}\sum_{k=1}^{N}   \overline{P}_{{\rm s},k} f_{{\rm s},k} \beta_k \bigg)\pi^{\left(\mathcal{S}_{\rm OD}\right)}_{{\rm p},\circ}\!\\&\ge\!\bigg(\overline{P}_{\rm s,sd}\!+\!P_{\rm s,sd}\sum_{k=1}^{N}   \overline{P}_{{\rm s},k} f_{{\rm s},k} \beta_k \bigg)\pi^{\left(\mathcal{S}_{\rm RD}\right)}_{{\rm p},\circ}\!=\! \mu^{\left(\mathcal{S}_{\rm RD}\right)}_{\rm s}.
   \label{eqn2}
  \end{split} \normalsize
\end{equation}
From (\ref{eqn3}) and (\ref{eqn2}), we have $\mu^{\left(\mathcal{S}_{\rm OD}\right)}_{\rm p}\ge \mu^{\left(\mathcal{S}_{\rm RD}\right)}_{\rm p}$ and $\mu^{\left(\mathcal{S}_{\rm OD}\right)}_{\rm s}\ge \mu^{\left(\mathcal{S}_{\rm RD}\right)}_{\rm s}$, and consequently, $\pi^{\left(\mathcal{S}_{\rm OD}\right)}_{{\rm p},\circ}\ge \pi^{\left(\mathcal{S}_{\rm RD}\right)}_{{\rm p},\circ}$ and $\pi^{\left(\mathcal{S}_{\rm OD}\right)}_{{\rm s},\circ}\ge \pi^{\left(\mathcal{S}_{\rm RD}\right)}_{{\rm s},\circ}$. The mean service rates of the relaying queues in $\mathcal{S}_{\rm OD}$ are lower bounded as
\begin{equation}
\small \begin{split}
    \mu^{\left(\mathcal{S}_{\rm OD}\right)}_{{\rm p},k}&\!=\!  \omega_k \pi^{\left(\mathcal{S}_{\rm OD}\right)}_{{\rm p},\circ}\pi^{\left(\mathcal{S}_{\rm OD}\right)}_{{\rm s},\circ} \alpha_{k}\overline{P}_{k,{\rm pd}}\\ &\!\ge\!  \omega_k \pi^{\left(\mathcal{S}_{\rm RD}\right)}_{{\rm p},\circ} \pi^{\left(\mathcal{S}_{\rm RD}\right)}_{{\rm s},\circ} \alpha_k\overline{P}_{k,{\rm pd}}\! =\!  \mu^{\left(\mathcal{S}_{\rm RD}\right)}_{{\rm p},k}
        \end{split} \normalsize
\end{equation}
%\omega_k \pi^{\left(\mathcal{S}_{\rm RD}\right)}_{{\rm p},\circ}\pi^{\left(\mathcal{S}_{\rm RD}\right)}_{{\rm s},\circ} \bigg[\!1\!-\! \overline{\alpha}_{k}{\rm Pr}\{Q_{{\rm s},k}\!\ne\! 0\}\!\bigg]\overline{P}_{k,{\rm pd}}=
     \begin{equation}
\small \begin{split}
      \mu^{\left(\mathcal{S}_{\rm OD}\right)}_{{\rm s},k}&\!=\!  \omega_k \pi^{\left(\mathcal{S}_{\rm OD}\right)}_{{\rm p},\circ} \pi^{\left(\mathcal{S}_{\rm OD}\right)}_{{\rm s},\circ} \big(\! 1\!-\! \alpha_{k}\!\big)\overline{P}_{k,{\rm sd}}\\& \!\ge\! \omega_k  \pi^{\left(\mathcal{S}_{\rm RD}\right)}_{{\rm p},\circ}  \pi^{\left(\mathcal{S}_{\rm RD}\right)}_{{\rm s},\circ} \big(\! 1\!-\! \alpha_{k}\!\big)\overline{P}_{k,{\rm sd}} \!=\!  \mu^{\left(\mathcal{S}_{\rm RD}\right)}_{{\rm s},k}.
    \end{split} \normalsize
\end{equation}
%\omega_k  \pi^{\left(\mathcal{S}_{\rm RD}\right)}_{{\rm p},\circ}  \pi^{\left(\mathcal{S}_{\rm RD}\right)}_{{\rm s},\circ} \bigg[\! 1\!-\! \alpha_{k}{\rm Pr}\{Q_{{\rm p},k}\ne 0\}\!\bigg]\overline{P}_{k,{\rm sd}}\!=\!

Since all the mean service rates of the queues in $\mathcal{S}_{\rm OD}$ are greater than or equal to the mean service rates of the queues in $\mathcal{S}_{\rm RD}$, the strategy $\mathcal{S}_{\rm OD}$ outperforms $\mathcal{S}_{\rm RD}$. Based on the above proof, setting $ \epsilon^{\left(\rm p\right)}_{1k}=\epsilon^{\left(\rm s\right)}_{1k}= \beta_k$ makes $\mathcal{S}_{\rm OD}$ outperform $\mathcal{S}_{\rm RD}$. Therefore, if we optimize over $ \epsilon^{\left(\rm p\right)}_{1k}$, $\epsilon^{\left(\rm s\right)}_{1k}$ and the remaining parameters of strategy $\mathcal{S}_{\rm OD}$, of course, we can get much higher performance than setting $ \epsilon^{\left(\rm p\right)}_{1k}=\epsilon^{\left(\rm s\right)}_{1k}= \beta_k$.
%if $\epsilon_{1k}=\beta^*_k$ and the other parameters of $\mathcal{S}_{\rm OD}$ are set to their equivalent optimal parameters in $\mathcal{S}_{\rm RD}$.
%If we optimize over $\epsilon_{1k}$\footnote{Instead of setting $ \epsilon^{\left(\rm p\right)}_{1k}=\epsilon^{\left(\rm s\right)}_{1k}= \beta_k$.} and the rest of $\mathcal{S}_{\rm OD}$'s parameters, of course, we can get better performance.
As a corollary to this proof, the strategy $\mathcal{S}_{\rm OD}$ outperforms $\mathcal{S}_{\rm RR}$.
\end{proof}

\subsection{The Case of Sensing Errors}
\begin{proof}
As mentioned in Section \ref{sensingerror}, the service rate of user $\ell$ and the arrival rate of the relaying queue of that user are reduced, on the average, by $\sum_r \omega_r \mathcal{B}^{\left(\ell r\right)}$ relative to the case of perfect sensing. In addition, the service rate of the $k${\it th} relaying queue of user $\ell$ is reduced by $\big(1-P_{\rm FA}^{k}\big)^2$ relative to the case of perfect sensing. Therefore, following the same steps as in the proof in the previous subsection, the service rates of the queues under $\mathcal{S}_{\rm RD}$ cannot exceed those in $\mathcal{S}_{\rm OD}$. Furthermore, the arrival rates of the relaying queues in $\mathcal{S}_{\rm RD}$ cannot be smaller than those in $\mathcal{S}_{\rm OD}$. Consequently, $\mathcal{S}_{\rm OD}$ outperforms $\mathcal{S}_{\rm RD}$. A direct result of this proof is that strategy $\mathcal{S}_{\rm OD}$ outperforms $\mathcal{S}_{\rm RR}$.
\end{proof}

%\section{}

%\section*{Acknowledgment}


%\ifCLASSOPTIONcaptionsoff
%  \newpage
%\fi



%\bibliographystyle{IEEEtran}

%\balance
 \bibliographystyle{IEEEtran}
 \bibliography{IEEEabrv,relays}

%\begin{biography}[{\includegraphics[width=1in,height=1.25in,clip,keepaspectratio]{mshell}}]{Michael Shell}

%\end{IEEEbiography}


%\newpage



%\vfill


\balance


%%%%%%%%%%%%%%%%%%%%%%%%%%%%%
%%%%%%%%%%%%%%%%%%%%%%%%%%%%%%%
%%%%%%%%%%%%%%%%%%%%%%%%%%%%%%
%%%%%%%%%%%%%%%%%%%%%%%%%%%%%%
%%%%%%%%%%% FIGURES
%%%%%%%%%%%%%%%%%%%%%%%%%%%%%%
%%%%%%%%%%%%%%%%%%%%%%%%%%%%%%%%
%%%%%%%%%%%%%%%%%%%%%%%%%%%%%%%%%%%%
%%%%%%%%%%%%%%%%%%%%%%%%%%%%%%%%%%%%
\end{document}